\documentclass[11pt]{article}
\usepackage[utf8]{inputenc}
\usepackage[left=1in, right=1in, bottom=1.15in, top=1.1in]{geometry}

\usepackage{color}
\usepackage{soul}

\usepackage{MG}

\title{Tarski Fixed Points in Quasi-FPT Queries}

\author{
Xi Chen\\Columbia University\\\url{xichen@cs.columbia.edu}
\and Ruiquan Gao\\Stanford University\\\url{ruiquan@cs.stanford.edu}
\and Yuhao Li\\Columbia University\\\url{yuhaoli@cs.columbia.edu}
\and Aviad Rubinstein\\Stanford University\\\url{aviad@cs.stanford.edu}
\and Mihalis Yannakakis\\Columbia University\\\url{mihalis@cs.columbia.edu}
}

\date{}

\begin{document}

\maketitle
\begin{abstract}
We study the query complexity of finding a Tarski fixed point over $[n]^k$. Previous work~has left a large gap between $\smash{{\Omega}(\log^2 n)}$ and $\smash{\log^{O(k)}n}$. We show that both of the previous  upper and lower bounds were far from tight: for every $k\geq 3$,
$$
\Omega\left((\log n)^{\frac{1}{2}\lceil \log k\rceil}\right)\le 
\Tarski(n,k)\le O\left(5^k (\log n)^{\lceil \log k\rceil}\right).$$
Succinctly, up to the fixed-parameter factor of $5^k$, the complexity is settled at $(\log n)^{\Theta(\log k )}$. In particular, we obtain the first super-polynomial query lower bound for this problem.
\end{abstract}

\thispagestyle{empty}
\newpage 
\setcounter{page}{1}

\newpage
\section{Introduction}

We begin by recalling the definition of monotone functions and
Tarski's fixed point theorem.

\begin{definition}[Monotone functions]
Let $(\calL,\preceq)$ be a complete lattice. A function
$f:\calL\rightarrow\calL$ is \emph{monotone} if
$f(x)\preceq f(y)$ for all $x,y\in\calL$ with $x\preceq y$.
\end{definition}

\begin{theorem}[\cite{tarski1955lattice}]\label{thm:tarski-fixed-point}
Let $(\calL,\preceq)$ be a complete lattice and
$f:\calL\rightarrow\calL$ be a monotone function.
Then the set of \emph{fixed points}, i.e., $\{x\in\calL:f(x)=x\}$,
is a nonempty complete lattice under the order of
$\calL$. In particular, $f$ has a least fixed point and a greatest
fixed point.
\end{theorem}

Tarski's fixed point theorem has had a lasting influence on several fields such as mathematics,
economics, logic, and computer science. 
Applications to set theory and topology already appear in Tarski’s original paper~\cite{tarski1955lattice}.
In economics, the theorem provides an elegant foundation for analyzing systemic
risk in financial networks~\cite{eisenberg2001systemic}.
In logic, it underlies the semantics of fixed point logics such as the
modal $\mu$-calculus~\cite{K83}.
In computer science, it provides foundations for the semantics of
recursive programs~\cite{SS71} and for static program analysis through
abstract interpretation~\cite{CC77}.

In this paper, we study the problem of finding a Tarski fixed point over
the $k$-dimensional grid lattice $([n]^k,\preceq)$, where $[n]^k:=\{1,\ldots,n\}^k$  and $x\preceq y$ if $x_i\le y_i$ for all $i\in [k]$.
This setting captures a broad
range of computational problems. In game theory, these include finding
pure Nash equilibria in supermodular games, an important class of games
with strategic complementarities~\cite{topkis1979equilibrium,topkis1998supermodularity,milgrom1990rationalizability}.
In verification, the problems of solving parity games~\cite{emerson1991tree,CJKLS217},
mean-payoff games~\cite{zwick1996complexity}, and Condon's simple
stochastic games~\cite{condon1992complexity} all reduce to finding a
Tarski fixed point~\cite{EPRY20}. These games play a central role
in verification and synthesis; for example, solving parity games is
equivalent to model checking the modal
$\mu$-calculus~\cite{emerson_modelchecking_1993}.
The Tarski problem also subsumes \textsc{Arrival}, which asks
whether a train following deterministic switching rules on a directed
graph reaches its destination~\cite{DGKMW17,GHH21}.

For supermodular games, Etessami, Papadimitriou, Rubinstein, and Yannakakis~\cite{EPRY20} showed that finding a pure Nash equilibrium is
computationally equivalent to finding a Tarski fixed point over the grid lattice.
Parity games, mean-payoff games, simple stochastic games, and \textsc{Arrival} are also notable for their distinct complexity status: they are among the few natural problems in $\NP\cap\coNP$,
yet have resisted polynomial-time algorithms~\cite{Wed}.
Tarski fixed points provide a unified framework for several longstanding open problems, and a polynomial-time algorithm for the Tarski problem would therefore resolve all of them at once.

We study the \emph{query complexity} of finding a Tarski fixed point over $[n]^k$.
In the black-box~query model, an algorithm is given $n$ and $k$ and
oracle access to an unknown monotone function
$f:[n]^k$ $\rightarrow[n]^k$.
In each round, the algorithm chooses a query point $x\in[n]^k$ and receives the entire vector $f(x)$.
Queries may depend on previous answers, and the algorithm may perform arbitrarily expensive computation between queries. 
The goal is to find a fixed point of $f$ using as few queries as possible.
We use $\Tarski(n,k)$ to denote both this problem and its deterministic
query complexity, i.e., the minimum worst-case number of queries needed to
solve it. We also consider (Las Vegas) randomized algorithms, which are required to always find
a fixed point for every monotone
input function $f$. Their query complexity is the
worst-case expected number of queries. 

The query complexity of Tarski fixed points over $[n]^k$ has attracted
significant attention, yet remains poorly understood. On the lower bound
side, Etessami, Papadimitriou, Rubinstein, and Yannakakis~\cite{EPRY20}
proved an $\Omega(\log^2 n)$ lower bound already in two dimensions,
even for randomized algorithms. This was  generalized by Br{\^a}nzei, Phillips, and
Recker~\cite{branzei2025tarski} to $k$ dimensions, obtaining an $\tilde{\Omega}(k\log^2 n)$ lower bound.
On the upper bound side, Tarski's algorithm (value iteration) uses $O(nk)$ queries. The recursive {binary} search algorithm of
Dang, Qi, and Ye~\cite{dang2011computational} gives the baseline bound
of $O(\log^k n)$ queries.
Surprisingly, Fearnley, P{\'a}lv{\"o}lgyi, and Savani~\cite{FPS22}
improved this bound to $O(\log^2 n)$ in three dimensions.
Haslebacher and Lill~\cite{haslebacher2025levelset} subsequently gave
a different algorithm achieving the same bound.
More recently, Chen, Li, and Yannakakis~\cite{CLY26} obtained an
$O(\log^2 n)$-query algorithm in four dimensions.
These algorithms achieve optimal query upper    bounds in three and four
dimensions but are tailored to the geometry of these low-dimensional
settings. Their only known extensions to higher dimensions rely on
decomposition theorems~\cite{FPS22,CL22}, which still yield bounds of $\log^{O(k)} n$.
Finally the quantum query complexity of the Tarski fixed point problem has also recently
attracted attention~\cite{P26,LMYZ26}.\smallskip\medskip

\noindent\textbf{Our contributions.}
We significantly improve the best known upper and lower bounds for Tarski fixed points, showing that the upper bound of $\log^{O(k)}n$ and the lower bound of $\Omega(\log^2 n)$ were both far from tight.

On the upper-bound side, we give the first algorithm whose query complexity has only logarithmic dependence on $k$ in the exponent of $\log n$.\footnote{Throughout the paper, $\log$ denotes the base-two logarithm.}
\begin{theorem}\label{theorem: upper}
For $k\geq 3$, there is a deterministic algorithm
for $\Tarski(n,k)$ with query complexity
$$
  O\!\left(5^k(\log n)^{1+\lceil\log((k+2)/3)\rceil}\right)=O\!\left(5^k(\log n)^{\lceil\log k\rceil}\right).
$$
\end{theorem}

Our algorithm builds on the safe function framework developed earlier
in~\cite{CLY23,CLY26} but is otherwise self-contained. In particular,
our algorithm recovers the optimal \(O(\log^2 n)\)-query bounds in three and four dimensions within a unified framework that applies to
arbitrary dimensions, without relying on geometric properties specific to these low-dimensional settings.
We present our algorithm in \Cref{sec:upper-bounds}. The safe function framework is reviewed in \Cref{sec:upper-preliminaries} and then we give an overview of the main ideas behind the new algorithm in \Cref{sec:mainoverview}. 

On the lower bound side, we establish the first super-polynomial query lower bound for Tarski fixed points.
Specifically, we prove a lower bound of $(\log n)^{\Omega(\log k)}$ for randomized algorithms, even under the promise that the input function has a unique fixed point.
\begin{restatable}[]{theorem}{MainLower}
\label{theorem: lower}
For $k\geq2$, any randomized algorithm that solves
$\Tarski(n,k)$ {\em (}with probability $1${\em )} requires
$$
  \Omega\!\left(
    \left(\frac{\log n}{\polylog(k)}\right)^{\lfloor\log_4(6k+4)\rfloor}
  \right)
$$
queries in expectation in the worst case, even when the input function has a unique fixed point.
\end{restatable}

This rules out any (randomized) algorithm with query complexity
polynomial in $k$ and $\log n$. In particular, this implies that
$\Tarski\notin\mathrm{FP}^{\mathrm{dt}}$ in the decision-tree model of $\TFNP^{\mathrm{dt}}$~\cite{GHJ+24}. This is in sharp contrast to contraction fixed points under the $\ell_p$-norm, which admit polynomial-query algorithms for all $p\in[1,\infty)\cup\set{\infty}$~\cite{CLY24,haslebacher2025query}. This is also the first super-polynomial lower bound for Tarski fixed points on \emph{any} family of finite lattices in
the model where the lattice and its order are fully known and only
the monotone function is accessed through an oracle.

\renewcommand{\Fix}{\mathtt{Fix}}

\section{The Algorithm}\label{sec:upper-bounds}

We prove \Cref{theorem: upper} for $\Tarski$ in this section. We first recall definitions of two related problems $\TarskiStar$ and $\SafeTarskiStar$ in %
\Cref{sec:upper-preliminaries}.
We state our main theorem for $\SafeTarskiStar$ in  \Cref{sec:mainoverview} (\Cref{thm:star-recurrence}), from which  \Cref{theorem: upper} follows.
After an overview of the  algorithm behind \Cref{thm:star-recurrence},
 we develop the new ingredients in 
 \Cref{sec:query-slice-bounds,sec:badpair,sec:separating-slice,sec:coordinate-contraction,sec:candidate-elimination}
 and then combine them to finish the proof in \Cref{sec:together}.

\subsection{Preliminaries}\label{sec:upper-preliminaries}

\noindent\textbf{The $\TarskiStar$ problem.}
$\Tarski$ is closely related to a problem called $\Tarski^*$~\cite{CL22}.
The~input to $\Tarski^*(n,k)$ is a \emph{monotone} function $f:[n]^k\rightarrow [n]^k\times \{\pm 1\}$, meaning that $f(x)\preceq f(y)$ for any $x\preceq y$.
We will refer to $f(x)_{k+1}\in \{\pm 1\}$ as the \emph{label} of $x$. 
Note that monotonicity holds for all $k+1$ coordinates, including the labels. 
We recall the definition of $\Tarski^*$ as follows:

\begin{definition}[$\TarskiStar(n,k)$]\label{def:tarski-star}
  Given oracle access to a monotone function
  $f:[n]^k\rightarrow[n]^k\times\{\pm 1\}$, find a point $x\in[n]^k$
  such that either $(x,0)\preceq f(x)$ or $f(x)\preceq(x,0)$.
  We also write $\TarskiStar(n,k)$~to denote its deterministic query complexity.
  We set $\TarskiStar(n,0)=1$ by convention.
\end{definition}

A point $x\in[n]^k$ is called a \emph{positive solution} to $f$ if $(x,0)\preceq f(x)$, and a \emph{negative solution} to $f$ if $f(x)\preceq(x,0)$. We call $x$ a \emph{solution} to $f$ if it is either a positive or a negative solution.
The goal of $\Tarski^*$ is then to find a solution to the unknown monotone function $f$.

The standard binary search algorithm can reduce $\Tarski(n,k)$ to
$\TarskiStar(n,k-1)$ and give the following bound~\cite{FPS22,CL22}.
See~\cite{CLY26} for a sketch of the reduction.

\begin{lemma}[Folklore]\label{lem:tarski-star-reduction}
  For every $k\geq1$, we have
  $$
    \Tarski(n,k)=O\big((\log n)\cdot \TarskiStar(n,k-1)\big).
  $$
\end{lemma}

Previous work improved the query complexity of $\Tarski$ by developing
better algorithms for $\TarskiStar$. In particular, it was shown that
$\TarskiStar(n,2)=O(\log n)$~\cite{FPS22} and
$\TarskiStar(n,3)=O(\log n)$~\cite{CLY26}.
When combined with \Cref{lem:tarski-star-reduction}, these bounds yield optimal $O(\log^2 n)$-query algorithms for $\Tarski(n,3)$ and $\Tarski(n,4)$, respectively.\medskip\vspace{0.05cm}

\noindent\textbf{The $\SafeTarskiStar$ problem.}
The upper bound of $\Tarski^*(n,3)=O(\log n)$ was proved using an approach based on \emph{safe} functions, which are monotone functions that satisfy an additional \emph{safety} condition. We need the following notation to define safe functions.
Here a \emph{slice} of $[n]^k$ is specified by an $s\in([n]\cup\{*\})^k$.
We write $$\calF(s)=\big\{i\in[k]:s_i=*\big\}\quad\text{and}\quad \calL_s=\big\{x\in[n]^k:x_i=s_i\text{ for }i\in[k]\setminus\calF(s)\big\}$$
to denote its set of free coordinates and its set of points, respectively.
A point $x\in \calL_s$ is a fixed point of $f$ on slice $s$ if $f(x)_i=x_i$ for all $i\in \calF(s)$; $x\in [n]^k$ is a fixed point of $f$ if $f(x)_i=x_i$ for all $i\in [k]$.

\begin{definition}[Safe Functions]\label{def:safesignfunction}
A monotone function $f:[n]^k\rightarrow [n]^k\times \{\pm 1\}$ is said to be \emph{safe} if it satisfies the following \emph{safety} condition: $f$ has a \emph{unique} fixed point on \emph{every slice} $s$ of $[n]^k$.
\end{definition}

The approach based on safe functions uses the following $\SafeTarskiStar(n,k)$ problem, which is the same~as $\TarskiStar(n,k)$, except that $f$ is not  only monotone but also promised to be safe:

\begin{definition}[$\SafeTarskiStar(n,k)$]\label{def:safe-tarski-star}
  Given oracle access to a safe function
  $f:[n]^k\rightarrow[n]^k$ $\times\{\pm 1\}$, %
  find a point $x\in[n]^k$
  such that $x$ is a solution to $f$.
  We also write $\SafeTarskiStar(n,k)$ to denote its deterministic query complexity, and
  set $\SafeTarskiStar(n,0)=1$ by convention.
\end{definition}

It was shown in \cite{CLY23,CLY26} that $\SafeTarskiStar$ and $\TarskiStar$ have exactly the same query complexity so it suffices to give algorithms for $\SafeTarskiStar$:

\begin{theorem}[\cite{CLY23,CLY26}]\label{thm:safe-tarski-equivalence}
  For every $n,k\geq 1$, $\TarskiStar(n,k)=\SafeTarskiStar(n,k)$.
\end{theorem}

\subsection{Main theorem and proof overview}\label{sec:mainoverview}

The main theorem of this section is a recursive algorithm for $\SafeTarskiStar(n,k)$ that leads to the following recurrence:

\begin{theorem}\label{thm:star-recurrence}
  For every $k\geq2$, we have
  $$
    \SafeTarskiStar(n,k)
    =O\!\left(\big(k\,2^k\log n\big)\cdot
    \SafeTarskiStar\big(n,\lfloor k/2\rfloor-1\big)\right).
$$
\end{theorem}
Solving the recurrence in \Cref{thm:star-recurrence} and applying
\Cref{lem:tarski-star-reduction} and \Cref{thm:safe-tarski-equivalence}, we obtain \Cref{theorem: upper}.
We defer the detailed calculation to \Cref{app:upper-bound}.\medskip 

\def\ALG{\textsc{Alg}}

\noindent\textbf{Proof overview.} To prove \Cref{thm:star-recurrence}, we give an algorithm for $\SafeTarskiStar(n,k)$ in the rest of this section.
We first give an overview of the main technical ideas behind it.

Let $Q$ be the set of queries made so far by an algorithm on $\SafeTarskiStar$, together with~their answers.
Given such a query-answer set $Q$, the unknown function can be any safe function that is consistent with $Q$.
In \Cref{sec:query-slice-bounds}, we prove \Cref{lem:query-slice-bounds}, which gives lower and upper bounds on the unique fixed point in every slice $s$, for any safe function consistent with $Q$. Using these bounds,~we
give an explicit definition of a set called $\Fix(Q)\subseteq [n]^k$ and show that it always contains the unique fixed point of every safe function consistent with $Q$. We then incorporate the queried labels (i.e., information about coordinate $k+1$ in $Q$) to define two subsets, $\Fix^-(Q)$ and $\Fix^+(Q)$, of $\Fix(Q)$, consisting of points in $\Fix(Q)$ that can still receive labels $-1$ and $+1$, respectively.

A key technical lemma (\Cref{lem:separating-threshold}) proved in \Cref{sec:separating-slice} shows that if $\Fix^-(Q)$ and $\Fix^+(Q)$ are \emph{coordinate-separable} along one of the first $k$ coordinates, say coordinate $k$, in the sense that 
$$
\max_{x\in \Fix^-(Q)} x_k \le \min_{y\in \Fix^+(Q)} y_k,
$$
then we are very close to solving $\SafeTarskiStar(n,k)$: a recursive call to $\SafeTarskiStar(n,k-1)$ is enough to finish the job.
Inspired by \Cref{lem:separating-threshold}, the main goal of our algorithm is to force such a coordinate separation between $\Fix^-(Q)$ and $\Fix^+(Q)$ using few queries.

For this purpose we measure  our progress toward this goal using the number of \emph{bad point-pairs} $(x,y)$: $x\in \Fix^-(Q)$ and $y\in \Fix^+(Q)$ such that $x_k>y_k$. These are exactly the pairs obstructing coordinate separation along coordinate $k$.
We therefore seek to query-efficiently reduce the number of bad point-pairs. This is done  in \Cref{sec:coordinate-contraction,sec:candidate-elimination}: the former establishes the needed structural property of $\Fix^-(Q)$ and $\Fix^+(Q)$, while the latter gives an algorithmic procedure that exploits it.

More specifically, we prove \Cref{lem:heavy-families} in \Cref{sec:coordinate-contraction}, showing that there always exists a \emph{heavy antipodal set-pair} $(S^-,S^+)$ with $S^-\subseteq \Fix^-(Q)$ and $S^+\subseteq \Fix^+(Q)$ satisfying: 
\begin{flushleft}\begin{enumerate}
    \item The set-pair is  \emph{heavy}: Eliminating either $S^-$ from $\Fix^-(Q)$ or $S^+$ from $\Fix^+(Q)$ would remove at least a $1/2^k$ fraction of bad point-pairs\footnote{While we call these sets heavy, this $1/2^k$ factor is responsible for the $2^k$ loss in the recurrence  of \Cref{thm:star-recurrence}.}; and 
    \item The set-pair is \emph{antipodal}: There exists a threshold $t\in [n]^k$ such that, in every coordinate $i\in [k-1]$, all points in $S^-$ and all points in $S^+$ lie on opposite sides of $t_i$ (i.e., either all points in $S^-$ are at most $t_i$ and all points in $S^+$ are at least $t_i$, or vice versa), while in coordinate $k$, all points in $S^+$ are at most $t_k$ and all points in $S^-$ are at least $t_k+1$. (See item~(ii) of \Cref{lem:heavy-families} for the formal condition.)
\end{enumerate}\end{flushleft}

In \Cref{sec:candidate-elimination} we prove \Cref{lem:candidate-elimination}, which shows that  for any such antipodal set-pair $(S^-,S^+)$, one can eliminate at least one of the two sets using a recursive call to $\SafeTarskiStar(n,\ell)$ for some $\ell\le \lfloor k/2\rfloor -1$.
Combining \Cref{lem:heavy-families} and \Cref{lem:candidate-elimination}, we can eliminate at least a $1/2^k$ fraction of the bad point-pairs using one recursive call to $\SafeTarskiStar(n,\ell)$.
Repeating this process eventually eliminates all bad point-pairs, after which $\Fix^-(Q)$ and $\Fix^+(Q)$ are coordinate-separable along coordinate $k$, and  \Cref{lem:separating-threshold} can be applied.
Finally we combine these ingredients in \Cref{sec:together} to give an algorithm for $\SafeTarskiStar$ that proves \Cref{thm:star-recurrence}.

\subsection{Query-answer sets on safe functions}
\label{sec:query-slice-bounds}

A set $Q=\{(x^1,r^1),\ldots,(x^q,r^q)\}$, for some $q\ge 0$, is called a \emph{query-answer set} if $x^1,\ldots,x^q$ are~distinct points in $[n]^k$ and $r^i\in [n]^k\times \{\pm 1\}$ for each $i\in [q]$. We use a query-answer set $Q$ to record queries that an algorithm has made on $\SafeTarskiStar(n,k)$ during its first $q$ rounds, together with their answers from the oracle.
(So $Q=\varnothing$ at the beginning.)
We say a query-answer set $Q$ and a function $h:[n]^k\rightarrow[n]^k\times\{\pm1\}$ are \emph{consistent}, denoted by $h\Rightarrow Q$, if $h(x)=r$ for all $(x,r)\in Q$.

Given that the unknown function in $\SafeTarskiStar(n,k)$ is promised to be safe,
we will only be interested in query-answer sets $Q$ that are consistent with at least one safe function $h$. We refer to such a $Q$ as a \emph{safe} query-answer set.
Given a slice $s$ of $[n]^k$,
our goal in this subsection is to use information in~$Q$ to obtain lower and upper bounds on the 
unique fixed point on slice $s$ of any safe function $h\Rightarrow Q$ (\Cref{lem:query-slice-bounds}). These bounds will be the starting point for defining our candidate
fixed points $\Fix(Q)$ in the next subsection. 
We note that this subsection is all about the first $k$ coordinates; label information of coordinate $k+1$ in $Q$ will only be used later in the next subsection when we define the two subsets $\Fix^-(Q)$ and $\Fix^+(Q)$ of $\Fix(Q)$.

We start with some definitions.
We say a point $y\in\calL_s$ is \emph{postfixed} for $Q$
on slice $s$ if $h(y)_i\geq y_i$ for every $i\in\calF(s)$ and every monotone\footnote{We use monotone rather than safe functions here for simplicity. Restricting to safe functions can only tighten~the lower and upper bounds in \Cref{lem:query-slice-bounds}, but it is not clear how to exploit this  to improve the query complexity.}
function $h$ consistent with $Q$. Similarly, a point $y\in \calL_s$ is \emph{prefixed}
for $Q$ on $s$ if $h(y)_i\leq y_i$ for every $i\in\calF(s)$ and every
monotone function $h$ consistent with $Q$.
When $s=(*,\ldots,*)$, we simply say that $y$ is postfixed or prefixed for $Q$.

We prove the following lemma on the sets of postfixed and prefixed points: %

\begin{lemma}\label{lem:hehe8}
Let $Q$ be a safe query-answer set and $s$ be a slice of $[n]^k$.
The set of postfixed (prefixed) points for $Q$ on slice $s$ is nonempty and closed under coordinatewise joins (or meets, respectively).
\end{lemma}
\begin{proof}
Let $I=\calF(s)$, and let $\ell,u\in\calL_s$ be the bottom and top
points of the slice, respectively. That is, for every $i\in[k]$, $\ell_i=1$ if $i\in I$ and $\ell_i=s_i$ if $i\notin I$, 
and $u_i=n$ if $i\in I$ and $u_i=s_i$ if $i\notin I$.

For every monotone function $h\Rightarrow Q$, the
boundary conditions imply $h(\ell)_i\geq\ell_i$ and $h(u)_i\leq u_i$ for
every $i\in I$. So $\ell$ is
postfixed and $u$ is prefixed for $Q$ on $s$, proving that both sets
are nonempty.

Next, let $a,b\in\calL_s$ be postfixed for $Q$ on $s$, and set
$c=a\vee b$, where $c_i=\max\{a_i,b_i\}$ for every $i\in[k]$.
Since $a$ and $b$ agree in all fixed coordinates, $c\in\calL_s$.
Fix any monotone function $h$ consistent with $Q$ and any
$i\in I$. By interchanging $a$ and $b$ if necessary, we assume
without loss of generality that $c_i=a_i$. As $a\preceq c$, monotonicity of
$h$ gives
$$
  h(c)_i\geq h(a)_i\geq a_i=c_i,
$$
where the second inequality follows because $a$ is postfixed for
$Q$ on $s$. 
This shows that $c$ is postfixed for $Q$ on $s$ and thus, the set of postfixed points for $Q$ on $s$ is closed under coordinatewise joins.
It can be proved similarly that the set of prefixed points is closed under coordinate meets. 
\end{proof}

Given \Cref{lem:hehe8}, we know that there is a \emph{greatest} postfixed point for $Q$ on slice $s$, denoted by $J_s(Q)$, and there
is a \emph{least} prefixed point, denoted by $M_s(Q)$.
For the all-$*$ slice, we denote  these two points as $J(Q)$ and $M(Q)$ for convenience. 

The next lemma shows that the unique fixed point $x^*$ on slice $s$ of any safe function consistent with $Q$ must lie between $J_s(Q)$ and $M_s(Q)$:

\begin{lemma}\label{lem:query-slice-bounds}
Let $Q$ be a safe query-answer set, and $h$ be any safe function consistent with $Q$.
Then for any slice $s$ of $[n]^k$, the unique fixed point $x^*$ of $h$ on slice $s$ satisfies  $J_s(Q)\preceq x^* \preceq M_s(Q)$.
This in particular implies that $J_s(Q)\preceq M_s(Q)$.
\end{lemma}

\begin{proof}
Fix a slice $s$, and let $I=\calF(s)$. Define a map
$f_s:\calL_s\rightarrow\calL_s$ by setting
$$
  f_s(x)_i=\begin{cases}
    h(x)_i&i\in I\\
    s_i&i\notin I
  \end{cases}\ , \quad\text{for every
$x\in\calL_s$ and $i\in[k]$.}
$$
It is easy to verify that $f_s$ maps $\calL_s$ to $\calL_s$ and is monotone, i.e., $f_s(x)\preceq f_s(y)$ for any $x\preceq y\in \calL_s$.
Moreover, a point $x\in\calL_s$ is a fixed point of $f_s$ if and
only if $h(x)_i=x_i$ for every $i\in I$. Since $h$ is safe, $x^*$ is the unique fixed point of $f_s$.

To prove the lower bound $J_s(Q)\preceq x^*$, let $x^0=J_s(Q)$ and define
$x^{t+1}=f_s(x^t)$ for $t=0,1,\ldots$. Since $x^0$ is
postfixed for $Q$ on $s$ and $h$ is consistent with $Q$, we have
$h(x^0)_i\geq x^0_i$ for all $i\in I$ and thus, $x^0\preceq f_s(x^0)=x^1$. By an induction, it is easy to show that  
  the sequence of points $x^0,x^1,\ldots$ is nondecreasing. Since all its points
belong to the finite set $\calL_s$, there is an integer $T\geq0$
such that $x^T=x^{T+1}=f_s(x^T)$. Hence $x^T$ is a fixed point of
$f_s$, so $x^T=x^*$ by uniqueness. We conclude that
$J_s(Q)=x^0\preceq x^T=x^*$.
The upper bound $x^*\preceq M_s(Q)$ can be proved similarly.
\end{proof}

Before moving on, we record the following simple fact that, as an algorithm makes more queries and extends the query-answer set $Q$, $J_s(Q)$ can only increase and $M_s(Q)$ can only decrease:

\begin{lemma}\label{lem:haha3}
Let $Q\subseteq Q'$ be two safe query-answer sets. Then $J_s(Q)\preceq J_s(Q')$ and $M_s(Q')\preceq M_s(Q)$ for every slice $s$ of $[n]^k$.    
\end{lemma}
\begin{proof}
Fix a slice $s$.
Given that any monotone function consistent with $Q'$ must also be consistent with $Q$, the set of postfixed points for $Q$ on $s$ is a subset of that for $Q'$ on $s$. It then follows that $J_s(Q)\preceq J_s(Q')$. The proof for prefixed points is similar.
\end{proof}

\subsection{Candidate fixed points}\label{sec:badpair}

Let $Q$ be a safe query-answer set.
Given \Cref{lem:query-slice-bounds}, it is natural to define the following set
$$
  \Fix(Q)=\left\{x\in[n]^k:
    J_s(Q)\preceq x\preceq M_s(Q)\text{ for every slice }s
    \text{ with }x\in \calL_s
  \right\}.
$$
For any safe function $g\Rightarrow Q$, let $x^*$ be its unique fixed point in $[n]^k$. Then $x^*$ must be the unique fixed point on every slice $s$ with $x^*\in \calL_s$. It follows from \Cref{lem:query-slice-bounds} that $x^*\in \Fix(Q)$.
This is why we will refer to $\Fix(Q)$ as our set of
\emph{candidate fixed points} given $Q$, since it contains the unique fixed point of every safe function $g\Rightarrow Q$.

The next two lemmas give more properties about $\Fix(Q)$.
First we show that $\Fix(Q)$ always contains $J(Q)$ and $M(Q)$.

\begin{lemma}\label{lem:candidate-containment}
For any safe query-answer set $Q$, we have $J(Q)\in \Fix(Q)$ and $M(Q)\in \Fix(Q)$.
\end{lemma}
\begin{proof}
Let $a=J(Q)$, and fix any slice $s$ with $a\in \calL_s$. We will show that $a=b:=J_s(Q)$.

On the one hand, since $a$ is
postfixed for $Q$ on the whole grid, it is also postfixed on $s$.
Thus~we have $a\preceq b$. On the other hand, to see $b\preceq a$, fix any monotone function $h\Rightarrow Q$. For each 
$i\in\calF(s)$, the definition of $b$ gives $h(b)_i\geq b_i$.
For each $i\notin\calF(s)$, we have $b_i=a_i$, and
monotonicity gives
$$
  h(b)_i\geq h(a)_i\geq a_i=b_i.
$$
So $b$ is postfixed for $Q$ on the whole grid. By the definition
of $J(Q)$, we have $b\preceq a$, so $J_s(Q)=a$.

By \Cref{lem:query-slice-bounds}, we now have
$J_s(Q)=a\preceq M_s(Q)$. Since this holds for every slice $s$ with $a\in \calL_s$, it follows from definition that $a\in\Fix(Q)$. The argument for $M(Q)$ is
symmetric.
\end{proof}

The next lemma shows that any point $x$, at which $Q$ forces a strict inequality of $h(x)_i>x_i$ or $h(x)_i<x_i$ for any $i\in[k]$, must be excluded from $\Fix(Q)$. The proof uses a one-dimensional slice and the safety condition to obtain a bound that excludes the point:

\begin{lemma}\label{lem:haha1}
Let $Q$ be a safe query-answer set. If $x\in [n]^k$ and $i\in [k]$ are such that $h(x)_i>x_i$
for every monotone $h\Rightarrow Q$, or such that $h(x)_i<x_i$ for every monotone $h\Rightarrow Q$, 
then $x\notin \Fix(Q)$.
\end{lemma}

\begin{proof}
Let $x$ and $i\in [k]$ be such that $h(x)_i>x_i$ for every monotone function $h\Rightarrow Q$.
Since these functions take values in $[n]^k\times\{\pm1\}$, we have
$x_i<n$. Let $s$ be the one-dimensional slice through $x$ with free
coordinate $i$, and set $v=x+\mathbf e_i$. By monotonicity, $h$ must satisfy 
$$
  h(v)_i\geq h(x)_i\geq x_i+1=v_i.
$$
So $v$ is postfixed for $Q$ on slice $s$ and thus, $v\preceq J_s(Q)$.
Since $v_i>x_i$, the condition $J_s(Q)\preceq x$ fails, and hence
$x\notin\Fix(Q)$.
The case when $h(x)_i<x_i$ for every monotone $h\Rightarrow Q$ is symmetric.
\end{proof}

Next we incorporate  the label information of $Q$ in coordinate $k+1$ to define subsets $\Fix^-(Q)$ and $\Fix^+(Q)$. We define them as follows: %
\begin{align*}
  \Fix^-(Q)&=\big\{x\in\Fix(Q):
    \nexists\hspace{0.04cm}(y,r)\in Q\text{ with }y\preceq x,\ r_{k+1}=+1\big\}\quad\text{and}\\[0.4ex]
  \Fix^+(Q)&=\big\{x\in\Fix(Q):
    \nexists\hspace{0.04cm}(y,r)\in Q\text{ with }x\preceq y,\ r_{k+1}=-1\big\}.
\end{align*}
So in $\Fix^-(Q)$ (or $\Fix^+(Q)$) we exclude $x\in \Fix(Q)$ that must be labelled $+1$ (or $-1$, respectively) in any safe $g\Rightarrow Q$. 
Note that their union is $\Fix(Q)$ but they do not form a partition of $\Fix(Q)$ in general. 
It is intuitive to see that, as more queries are made, all three sets can only shrink:
\begin{lemma}\label{lem:haha4}
If $Q\subseteq Q'$ are safe query-answer sets, then we have
$$
  \Fix(Q')\subseteq\Fix(Q),\quad
  \Fix^-(Q')\subseteq\Fix^-(Q)\quad\text{and}\quad
  \Fix^+(Q')\subseteq\Fix^+(Q).
$$
\end{lemma}
\begin{proof}
Let $x\in\Fix(Q')$. For every slice $s$ containing $x$,
\Cref{lem:haha3} gives
$$
  J_s(Q)\preceq J_s(Q')\preceq x
  \preceq M_s(Q')\preceq M_s(Q).
$$
Thus $x\in\Fix(Q)$, and we have proved that $\Fix(Q')\subseteq \Fix(Q)$. 

Moreover, since
$Q\subseteq Q'$, any queried label in $Q$ that excludes a point from
$\Fix^-(Q)$ or $\Fix^+(Q)$ also excludes it from the corresponding
set for $Q'$. This proves the other two inclusions.
\end{proof}

So as an algorithm makes more queries, all three sets can only shrink. While $\Fix(Q)$ can never be empty, 
the last lemma of this subsection shows that, if one of $\Fix^-(Q)$  and $\Fix^+(Q)$ becomes empty, a solution to the unknown function has already been revealed by $Q$:

\begin{lemma}\label{lem:haha5}
If $\Fix^-(Q)=\varnothing$, then $J(Q)$ must be a positive solution to every safe function $g\Rightarrow Q$. 

If $\Fix^+(Q)=\varnothing$, then $M(Q)$ must be a negative solution to every safe function    $g\Rightarrow Q$.
\end{lemma}
\begin{proof}
Suppose that $\Fix^-(Q)=\varnothing$, and let $a=J(Q)$.
By \Cref{lem:candidate-containment}, $a\in\Fix(Q)$. Given that~$a\notin\Fix^-(Q)$ (as $\Fix^-(Q)=\varnothing$), there exists $(y,r)\in Q$ with
$y\preceq a$ and $r_{k+1}=+1$.

Fix any safe function $g\Rightarrow Q$. Monotonicity gives
$g(a)_{k+1}=+1$. Since $a$ is postfixed for $Q$ on the whole grid,
we have that $a$ is a positive solution to $g$.
The other case is symmetric.
\end{proof}

\subsection{Dimension reduction via %
coordinate separation}\label{sec:separating-slice}

Given \Cref{lem:haha4} and \Cref{lem:haha5}, the high-level strategy of our algorithm for $\SafeTarskiStar$ is to make  queries to shrink $\Fix^-(Q)$ and $\Fix^+(Q)$, aiming to eliminate one of them with as few queries as possible. A key idea behind our new algorithm in this paper is to not just care about individual sizes of the two sets, but to aim at making the two sets \emph{coordinate-separable}: We say a pair of sets $(A,B)$, where $A,B\subseteq [n]^k$, is \emph{coordinate-separable}\footnote{Note that the definition uses specifically coordinate $k$. This is without loss of generality; the same proof works if we replace $k$ by any coordinate $i\in [k]$.}
 if %
$$\max_{x\in A} x_k\le \min_{y\in B} y_k.$$
Equivalently this means that there exists a $t\in [n]$ such that $x_k\le t\le y_k$ for all $x\in A$ and $y\in B$.

The following lemma shows that if $(\Fix^-(Q),\Fix^+(Q))$ is coordinate-separable, then indeed we are very close to eliminating one of them completely, which by \Cref{lem:haha5} will solve $\SafeTarskiStar$:

\begin{lemma}\label{lem:separating-threshold}
Let $Q$ be a safe query-answer set such that  $(\Fix^-(Q),\Fix^+(Q))$ is coordinate-separable.
There is an algorithm that, given $Q$ and oracle access to any safe function $g\Rightarrow Q$, makes at most $\SafeTarskiStar(n,k-1)+1$ additional queries to find a solution to $g$.
\end{lemma}
\begin{proof}
If either $\Fix^-(Q)$ or $\Fix^+(Q)$ is empty, the claim is trivial by \Cref{lem:haha5}. Below we assume that they are nonempty and, for some $t\in [n]$, satisfy
$$
  \max_{x\in\Fix^{-}(Q)}x_k\leq t\leq\min_{y\in\Fix^{+}(Q)}y_k.
$$
Given a safe function $f\Rightarrow Q$, we focus on its $(k-1)$-dimensional slice $s=(*,\ldots,*,t)$ and induce the following function 
$f':\calL_s\rightarrow \calL_s\times \{\pm 1\}$: For $i\in [k-1]$, we have $f'(z)_i=f(z)_i$; $f'(z)_k=t$; and its label in coordinate $k+1$ is defined by
\begin{equation}\label{eq:lex-label}
 f'(z)_{k+1}=+1
  \quad\Longleftrightarrow\quad
  f(z)_{k}>t
  \quad\text{or}\quad
  \bigl[\hspace{0.05cm}f(z)_k=t\ \text{and}\ f(z)_{k+1}=+1\hspace{0.05cm}\bigr].
\end{equation}
From now on, we view $f'$ as a function from $[n]^{k-1}$ to $[n]^{k-1}\times \{\pm 1\}$, by skipping coordinate $k$.
It can be easily verified that $f'$ must be safe given that $f$ is safe. Hence $f'$ is a valid input function~for $\SafeTarskiStar(n,k-1)$. Furthermore, given the definition of $f'$ from $f$, oracle access to $f'$ can be simulated by oracle access to $f$, query by query.

Given $Q$ and oracle access to any safe function $f\Rightarrow Q$, the algorithm proceeds as follows: It runs an algorithm $\ALG$ for $\SafeTarskiStar(n,k-1)$ on the unknown function $f'$ induced by $f$. Each time when $\ALG$ makes a query $z\in [n]^{k-1}$ about $f'$, our algorithm makes a query $y$ about~$f$, where $y$ is obtained from $z$ by inserting $t$ as its coordinate $k$.
After receiving $r=f(y)$ from the oracle, our algorithm updates $Q$ by adding a new pair $(y,r)$, replies to $\ALG$'s query (based on the construction of $f'$ from $f$), and continues.
This process stops within $\SafeTarskiStar(n,k-1)$ rounds; at the end, $\ALG$ returns a point $z\in [n]^{k-1}$ as a solution to $f'$. Let $y$ be its corresponding point in $[n]^k$. Our algorithm makes one additional query on $y$ (if it has never been queried before). This makes sure that at the end, letting $Q'\supseteq Q$ be the updated query-answer set, $Q'$ must contain a pair $(y,r)$ such that $y\in \calL_s$ and one of the following two conditions (with four subcases) holds:
\begin{enumerate}
\item $r_i\ge y_i$ for every $i\in [k-1]$ and either (1.1) $r_k>t$ or (1.2) $r_k=t$ and $r_{k+1}=+1$; or
\item $r_i\le y_i$ for every $i\in [k-1]$ and either (2.1) $r_k<t$ or (2.2) $r_k=t$ and $r_{k+1}=-1$.
\end{enumerate}
We show below that, in each of these four subcases, a solution has already been revealed in $Q'$.

In Case 1.2, it is clear that $y$ must be a positive solution in any safe function $f\Rightarrow Q'$;
in Case 2.2, $y$ must be a negative solution in any safe function $f\Rightarrow Q'$.

Consider Case 1.1. 
Let $y'\in [n]^k$ be obtained from $y$ by increasing
coordinate $k$ by one. Since $h(y)=r$ for every monotone
function $h\Rightarrow Q'$, monotonicity implies that $y'$ is
postfixed for $Q'$ over the whole grid.
Every $x\in\Fix(Q')$ satisfies $y'\preceq J(Q')\preceq x$, so $x_k\geq t+1$.
As every $x\in\Fix^{-}(Q)$ satisfies $x_k\leq t$ and
$\Fix^{-}(Q')\subseteq\Fix^{-}(Q)$, we have
$\Fix^{-}(Q')=\varnothing$ and by \Cref{lem:haha5}, $J(Q')$ must be a positive solution to every safe function consistent with $Q'$.
The proof of Case 2.1 is symmetric. 
\end{proof}

Inspired by 
\Cref{lem:separating-threshold}, our goal in the rest of this section is to design a query-efficient strategy to force 
  $(\Fix^-(Q),\Fix^+(Q))$ to become coordinate-separable.
A natural way to measure the amount of work to be done is to count the number of \emph{bad point-pairs}. Here we say $(x,y)$ is a bad point-pair if $x\in \Fix^-(Q),y\in \Fix^+(Q)$ and $x_k>y_k$. We use $\calR(Q)$ to denote the set of bad point-pairs. 
Note that $(\Fix^-(Q),\Fix^+(Q))$ is coordinate-separable if and only if $\mathcal{R}(Q)=\varnothing$, so bad point-pairs are obstructions to the coordinate separation that we aim for. 
Given that by \Cref{lem:haha4}, both $\Fix^-(Q)$ and $\Fix^+(Q)$ only shrink after each query, so does $\mathcal{R}(Q)$.
Naturally our next goal is to come up with a strategy to shrink $\calR(Q)$ query-efficiently.

\subsection{Existence of a heavy antipodal set-pair}\label{sec:coordinate-contraction}

We start by proving \Cref{lem:heavy-families}, which shows the existence of a \emph{heavy antipodal} set-pair $(S^-,S^+)$, where $S^-\subseteq \Fix^-(Q)$ and $S^+\subseteq \Fix^+(Q)$.
Being \emph{heavy} means that eliminating either one of them would reduce the size of $\calR(Q)$ by at least a $1/2^k$ fraction.
The structure of the set-pair~being~antipodal, to be defined below, will be exploited in the next subsection, where we give an algorithm that can eliminate one of the two sets with few queries.

\begin{lemma}\label{lem:heavy-families}
Let $Q$ be a safe query-answer set with $\mathcal{R}(Q)\ne\varnothing$.
There exists a set-pair $(S^-,S^+)$ with  $S^{-}\subseteq\Fix^{-}(Q)$ and
$S^{+}\subseteq\Fix^{+}(Q)$ such that the following two conditions hold:
\begin{flushleft}\begin{enumerate}[label=(\roman*)]
  \item The set-pair $(S^-,S^+)$ is \emph{heavy}: We have 
$$
    \begin{aligned}
      \bigl|\{(x,y)\in\mathcal{R}(Q):x\in S^{-}\}\bigr|
      ,\ \bigl|\{(x,y)\in\mathcal{R}(Q):y\in S^{+}\}\bigr|
      &\geq \frac{|\mathcal{R}(Q)|}{2^k}.
    \end{aligned}
  $$  
\item The set-pair $(S^-,S^+)$ is \emph{antipodal}: There is a partition
$[k]=J\mathbin{\dot\cup}H$ with $k\in H$, and a vector $t\in[n]^k$
satisfying that, for every $x\in S^{-}$ and $y\in S^{+}$, we have
\begin{equation}\label{eq:antipodal-pattern}
    x_J\preceq t_J\preceq y_J
    \quad\text{and}\quad
    x_H\succeq t_H\succeq y_H.
  \end{equation}
  Moreover, for coordinate $k$, we have $x_k\geq t_k+1$ for every $x\in S^{-}$.  
\end{enumerate}\end{flushleft}
\end{lemma}
\begin{proof}
We first handle coordinate $k$. Order the pairs
$(x,y)\in\mathcal{R}(Q)$ by $y_k$, and let $q$ be the $y_k$-value
of a middle pair. Define
$$
  \begin{aligned}
    R^{-}=\big\{(x,y)\in\mathcal{R}(Q):y_k\geq q\big\}\quad\text{and}\quad
    R^{+}=\big\{(x,y)\in\mathcal{R}(Q):y_k\leq q\big\}.
  \end{aligned}
$$
Both sets contain at least half of the bad pairs because of the median choice of $q$.
Every $(x,y)\in R^{-}$ satisfies $x_k>y_k\geq q$, while every
$(x,y)\in R^{+}$ satisfies $y_k\leq q$.
Set $t_k=q$, $J=\varnothing$, and $H=\{k\}$.

We now restrict $R^{-}$ and $R^{+}$ one coordinate at a time.
For each $i\in[k-1]$, order the pairs in the current $R^{-}$ by $x_i$
and those in the current $R^{+}$ by $y_i$, and let $a,b$ be the
respective coordinate values of middle pairs.
If $a\leq b$, restrict $R^{-}$ to pairs with $x_i\leq a$ and
$R^{+}$ to pairs with $y_i\geq b$, set $t_i=a$, and put $i$ in $J$.
If $a>b$, restrict $R^{-}$ to pairs with $x_i\geq a$ and
$R^{+}$ to pairs with $y_i\leq b$, set $t_i=b$, and put $i$ in $H$.
In either case, each set keeps at least half of its pairs, and the
required inequalities hold in coordinate $i$. Further restrictions
preserve all inequalities established in earlier coordinates.

After processing all $k-1$ coordinates, we have a partition
$[k]=J\mathbin{\dot\cup}H$ with $k\in H$, and the final pair of sets satisfy
$|R^{-}|\geq2^{-k}|\mathcal{R}(Q)|$ and $|R^{+}|\geq2^{-k}|\mathcal{R}(Q)|$.
Define
$$
  S^{-}=\big\{x:(x,y)\in R^{-}\text{ for some }y\big\}
  \quad\text{and}\quad
  S^{+}=\big\{y:(x,y)\in R^{+}\text{ for some }x\big\}.
$$
By construction, these families satisfy \eqref{eq:antipodal-pattern},
and every $x\in S^{-}$ has $x_k\geq t_k+1$, proving item (ii).
Every pair in $R^{-}$ has its first endpoint in $S^{-}$, and every
pair in $R^{+}$ has its second endpoint in $S^{+}$.
Since both pair sets are subsets of $\mathcal{R}(Q)$, their size
bounds imply the two bounds in item (i).
\end{proof}

\subsection{Eliminating one set from an antipodal set-pair}\label{sec:candidate-elimination}

Now let $Q$ be a query-answer set and let $(S^-,S^+)$, where $S^-\subseteq \Fix^-(Q)$ and $S^+\subseteq \Fix^+(Q)$, be a set-pair that is antipodal with respect to $t\in [n]^k$ (in the sense of item (ii) in \Cref{lem:heavy-families}).
In this subsection we give an algorithm  that eliminates one of the sets, i.e., either $S^-$ from $\Fix^-(Q)$ or $S^+$ from $\Fix^+(Q)$, using few queries.

We first give a simple lemma about $\Fix(Q)$:

\begin{lemma}\label{lem:candidate-comparison}
Let $Q$ be a safe query-answer set and $z\in [n]^k$. For any $I\subseteq [k]$, the following hold:
\begin{enumerate}[label=(\roman*)]
  \item If $h(z)_i\ge z_i$ for every $i\in I$ in any monotone $h\Rightarrow Q$,  then
  $x\notin\Fix(Q)$ for any $x$ such that
  $$
    x_j\geq z_j\quad\text{for every }j\in[k]\setminus I
    \quad\text{and}\quad
    x_i<z_i\quad\text{for some }i\in I.
  $$

  \item If $h(z)_i\le z_i$ for every $i\in I$ in any monotone $h\Rightarrow Q$, %
  then
  $x\notin\Fix(Q)$ for any $x$ such that
  $$
    x_j\leq z_j\quad\text{for every }j\in[k]\setminus I
    \quad\text{and}\quad
    x_i>z_i\quad\text{for some }i\in I.
  $$
\end{enumerate}
\end{lemma}

\begin{proof}
For (i), let $x$ be a point that satisfies the specified conditions, and set $w=x\vee z$.
We define a slice $s$ by setting
$s_i= *$ if $ x_i<z_i$ and 
$s_i=x_i$ if $x_i\ge z_i$. So $\calF(s)\ne \varnothing$ and $\calF(s)\subseteq I$.
We also have $x,w\in\calL_s$ and $x\prec w$, and  $w_i=z_i$ for every $i\in\calF(s)$.
By assumption, we have $h(z)_i\ge z_i$ for all $i\in I$ and every
monotone function $h$ consistent with $Q$.
As $z\preceq w$ and $w_i=z_i$ for $i\in\calF(s)$, monotonicity of
$h$ gives $h(w)_i\geq w_i$ for every $i\in \calF(s)$.
Thus $w$ is postfixed on $s$, so $x\prec w\preceq J_s(Q)$, which by definition excludes $x$ from $\Fix(Q)$.
The proof of (ii) is symmetric. 
\end{proof}

Now we give the algorithm that eliminates one set from $(S^-,S^+)$:

\begin{lemma}\label{lem:candidate-elimination}
Let $k\ge 2$. Let $Q$ be a safe
query-answer set, and 
$S^{-}\subseteq\Fix^{-}(Q)$ and $S^{+}\subseteq\Fix^{+}(Q)$
  be~two nonempty sets such that $(S^-,S^+)$ is an antipodal set-pair in the sense of (ii) in \Cref{lem:heavy-families}. 
There is an algorithm that, given $Q$, $(S^-,S^+)$ and oracle access to any safe function $g\Rightarrow Q$, makes at most  $\SafeTarskiStar(n,\lfloor k/2\rfloor -1)+1$ additional queries to $g$ to extend $Q$ to  $Q'$ such that 
either 
$$
  S^{-}\cap\Fix^{-}(Q')=\varnothing
  \quad\text{or}\quad
  S^{+}\cap\Fix^{+}(Q')=\varnothing.
$$
\end{lemma}

\begin{proof}
First, consider the special case when $J=\varnothing$. Then we know that
$y\preceq t\preceq x$ for every $x\in S^{-}$ and $y\in S^{+}$.
The algorithm simply queries $t$ and let
$Q'=Q\cup\{(t,r)\}$ be the resulting query-answer set.
If $r_{k+1}=+1$, monotonicity forces positive labels at every
$x\in S^{-}$ and thus, $S^{-}\cap\Fix^{-}(Q')=\varnothing$.
If $r_{k+1}=-1$, monotonicity forces negative labels at every
$y\in S^{+}$ and thus,~$S^{+}\cap\Fix^{+}(Q')=\varnothing$.
In either case, one of the sets is eliminated.

We now consider the general case of $J\neq\varnothing$. The proof follows a flow similar to that of \Cref{lem:separating-threshold}.
We build an induced $\SafeTarskiStar$ instance
on a slice $s$ with $\min\{|J|,|H|\}-1$ free coordinates.
We show that after this induced $\SafeTarskiStar$ instance is solved, the extended query-answer set $Q'$ must have either eliminated $S^-$ from $\Fix^-(Q)$ or $S^+$ from $\Fix^+(Q)$. 
One subtlety to point out is that $\min\{|J|,|H|\}-1$ could be $0$. In this case it can be checked that the only thing the algorithm needs to do is to query $t$; after that the updated $Q'$ can achieve the goal.

We start with the case when $|J|\le |H|$.\medskip

\noindent\textbf{First case:} $0<|J|\leq|H|$.
We choose any $j\in J$ and define the slice
$s\in([n]\cup\{*\})^k$ by
$$
  s_\ell=
  \begin{cases}
    *&\ell\in J\setminus\{j\}\\
    t_\ell&\ell\in H\cup\{j\}
  \end{cases}.
$$
Given any safe function $f\Rightarrow Q$, we use $f$ to induce a function $f':\calL_s\rightarrow \calL_s\times \{\pm 1\}$ as follows:~For any $w\in \calL_s$ and any $i\in J\setminus \{j\}$, set $f'(w)_i=f(w)_i$; set $f'(w)_i=w_i$ for every $i\in H\cup \{j\}$; and for coordinate $k+1$, set the label $f'(w)_{k+1}$  by
$$
   f'(w)_{k+1}=+1
  \quad\Longleftrightarrow\quad
  f(w)_j>w_j=t_j
  \quad\text{or}\quad
  \big[\hspace{0.05cm}f(w)_j=t_j\ \text{and}\ f(w)_{k+1}=+1\hspace{0.05cm}\big].
$$
We view $f'$ as a function from $[n]^{|J|-1}$ to $[n]^{|J|-1}\times \{\pm 1\}$ (by renumbering $J\setminus \{j\}$ to $[|J|-1]$ and 
$k+1$ to $|J|$).
It is easy to verify that $f'$ is safe given $f$ is safe. So $f'$ is a valid  $\SafeTarskiStar(n,|J|-1)$ instance.
Moreover, oracle access to $f'$ can be simulated by oracle access to $f$ query by query.

Given $Q$ and oracle access to any safe function $f\Rightarrow Q$, our algorithm proceeds as follows: We run an algorithm $\ALG$ for $\SafeTarskiStar(n,|J|-1)$ on the unknown safe function $f'$ induced by $f$. Each time when $\ALG$ makes a query $z\in [n]^{|J|-1}$ about $f'$, our algorithm makes a query  about~$f$ on a point $y$ that corresponds to $z$ as in the construction of $f'$ from $f$.
After receiving $r=f(y)$ from the oracle, our algorithm updates $Q$ by adding a new pair $(y,r)$, replies to $\ALG$'s query (based on the construction of $f'$ from $f$), and continues.
This process stops within $\SafeTarskiStar(n,|J|-1)$ rounds; at the end, $\ALG$ returns a point $z\in [n]^{|J|-1}$ as a solution to $f'$. Let $y$ be its corresponding point in $[n]^k$. Our algorithm makes one additional query on $y$ (if it has not been queried before). This makes sure that at the end, letting $Q'\supseteq Q$ be the updated query-answer set, $Q'$ must contain a pair $(y,r)$ such that $y\in \calL_s$ and one of the following two conditions (with four subcases) holds:
\begin{enumerate}
\item $r_i\ge y_i$ for every $i\in J\setminus \{j\}$ and either (1.1) $r_j>t_j$ or (1.2) $r_j=t_j$ and $r_{k+1}=+1$; or
\item $r_i\le y_i$ for every $i\in J\setminus \{j\}$ and either (2.1) $r_j<t_j$ or (2.2) $r_j=t_j$ and $r_{k+1}=-1$.
\end{enumerate}
(We note that for the special case when $J=\{j\}$, the above can be achieved simply by querying~$t$. Setting $y=t$ satisfies one of the two conditions above.)
We show below that Cases (1.1) and (1.2) eliminate $S^-$ from $\Fix^-(Q')$, while Cases (2.1) and (2.2) eliminate $S^+$ from $\Fix^+(Q')$.

Consider Cases (1.1) and (1.2).
We first show that $y$ is postfixed in $J$ for $Q'$. 
To see this,~take any monotone function $h\Rightarrow Q'$. 
Since $(y,r)\in Q'$, we have 
  $h(y)_i \geq y_i$ for every
$i\in J\setminus\{j\}$, and also
$h(y)_j\geq y_j$ (whether it is (1.1) or (1.2)). So $y$ is postfixed in $J$ for $Q'$.
Assume for a contradiction that $x\in S^{-}\cap\Fix^{-}(Q')$. Because
$x_H\succeq t_H=y_H$, \Cref{lem:candidate-comparison} with $I=J$
implies $y\preceq x$.
But given that $x_j\leq t_j=y_j$, we have $x_j=y_j$. We consider the two subcases separately below:
\begin{flushleft}\begin{enumerate}
\item \textbf{Case 1.1:}  $r_j>t_j$. Monotonicity forces $h(x)_j>x_j$ for every monotone
function $h\Rightarrow Q'$, contradicting with
$x\in\Fix(Q')$ by \Cref{lem:haha1}.
\item \textbf{Case 1.2:} $r_j=t_j$ and $r_{k+1}=+1$.
Monotonicity forces a positive label at $x$, contradicting
with $x\in\Fix^{-}(Q')$ by definition.
\end{enumerate}\end{flushleft}
So we conclude that $S^{-}\cap\Fix^{-}(Q')=\varnothing$.
Cases (2.1) and (2.2) can be handled similarly. 

\medskip

\noindent\textbf{Second case:}
$|H|<|J|$.
We define the slice $s\in([n]\cup\{*\})^k$ by
$$
  s_\ell=
  \begin{cases}
    *&\ell\in H\setminus\{k\} \\
    t_\ell &\ell\in J\cup\{k\} 
  \end{cases}.
$$
Given any safe function $f\Rightarrow Q$, we use $f$ to induce a function $f':\calL_s\rightarrow \calL_s\times \{\pm 1\}$ as follows:~For any $w\in \calL_s$ and $i\in H\setminus \{k\}$, set $f'(w)_i=f(w)_i$; set $f'(w)_i=w_i=t_i$ for every $i\in J\cup \{k\}$; and for coordinate $k+1$, set the label $f'(w)_{k+1}$  by
$$
  f'(w)_{k+1}=+1
  \quad\Longleftrightarrow\quad f(w)_k>w_k=t_k.
$$
We view $f'$ as a function from $[n]^{|H|-1}$ to $[n]^{|H|-1}\times \{\pm 1\}$ (by renumbering $H\setminus \{k\}$ to $[|H|-1]$ and 
$k+1$ to $|H|$).
It is easy to verify that $f'$ is safe given $f$ is safe, so $f'$ is a valid  $\SafeTarskiStar(n,|H|-1)$ instance.
Similarly, oracle access to $f'$ can be simulated by oracle access to $f$, query by query.

Our algorithm runs in a similar fashion by simulating an algorithm on $\SafeTarskiStar(n, |H|-1)$.
Within $\SafeTarskiStar(n, |H|-1)+1$ many queries, the updated query-answer set $Q'$ must contain a pair $(y,r)$ such that $y\in \calL_s$ and one of the following two conditions holds:
\begin{enumerate}
    \item $r_i\ge y_i$ for every $i\in H\setminus \{k\}$ and $r_k>t_k$; or
    \item $r_i\le y_i$ for every $i\in H$.
\end{enumerate}
Again note that for the special case when $H=\{k\}$, the above can be achieved easily by querying $t$ and setting $y=t$.
We show below that Case 1 eliminates $S^+$, while Case 2 eliminates $S^-$.

In Case 1, $y$ is postfixed in $H$ for $Q'$. Assume for a contradiction that $x\in S^+\cap \Fix^+(Q')$.
Since $x_J\succeq t_J=y_J$,
\Cref{lem:candidate-comparison} with $I=H$
implies $y\preceq x$.
But we also have $x_k\le t_k=y_k$, so $x_k=y_k$. 
Monotonicity then forces $h(x)_k>y_k=t_k$ for every monotone
function $h\Rightarrow Q'$, contradicting with $x\in\Fix(Q')$
by \Cref{lem:haha1}.
As a result, we have $S^{+}\cap\Fix^{+}(Q')=\varnothing$.

In Case 2, $y$ is  prefixed in $H$ for $Q'$.
But every $x\in S^{-}$ has
$x_J\preceq t_J=y_J$ and  $x_k\geq t_k+1>y_k$.
Thus \Cref{lem:candidate-comparison} with $I=H$ excludes every
$x\in S^{-}$ from $\Fix(Q')$, giving 
$S^{-}\cap\Fix^{-}(Q')=\varnothing$.
\end{proof}

\subsection{Putting everything together}\label{sec:together}
Now we are ready to prove \Cref{thm:star-recurrence}.
\begin{proof}[Proof of \Cref{thm:star-recurrence}]
We let $Q$ denote the safe query-answer set maintained by our algorithm.~We initialize $Q=\varnothing$ and proceed in rounds while  $\calR(Q)\ne\varnothing$; the algorithm moves to the next phase once $\calR(Q)$ becomes empty.  
To keep reducing the size of $\calR(Q)$, we apply 
  \Cref{lem:heavy-families} to obtain a heavy antipodal set-pair $(S^-,S^+)$.
Applying \Cref{lem:candidate-elimination} (and using that $k\ge 2$), 
  we can spend at most $\SafeTarskiStar(n,\lfloor k/2\rfloor-1)+1$ queries to reduce the size of $\calR(Q)$ by at least a $1/2^k$ fraction.
Given that we trivially have $|\calR(Q)|\le n^{2k}$ initially and queries can only shrink $\calR(Q)$, repeating this for at most $\lfloor 2k2^k \log n\rfloor$ rounds suffices to eliminate all bad point-pairs. 

Once $\calR(Q)$ becomes empty,
if either $\Fix^-(Q)$ or $\Fix^+(Q)$ is empty, we are already done using \Cref{lem:haha5}.  Otherwise, the sets are nonempty and coordinate-separable, so we can apply 
\Cref{lem:separating-threshold} to solve the rest of the problem with no more than $\SafeTarskiStar(n,k-1)+1$ queries. 

Hence, for $k\ge 2$ we get the following recurrence:
$$
  \SafeTarskiStar(n,k)
  \leq \SafeTarskiStar(n,k-1)+1
  +2k2^k\log n\Big(\SafeTarskiStar(n,\lfloor k/2\rfloor-1)+1\Big).
$$
Expanding the first term through dimensions $k-1,\ldots,2$
and
using the fact that $\SafeTarskiStar(n,d)$ is nondecreasing in $d$,
 gives a total query cost of at most
$$
\SafeTarskiStar(n,1)+( k-1)+2\log n\sum_{d=2}^k d2^d
\Big(\SafeTarskiStar(n,\lfloor k/2\rfloor-1)+1\Big).
$$
Since $\SafeTarskiStar(n,1)=O(\log n)$ and
$\sum_{d=2}^k d2^d=(k-1)2^{k+1}$, we have 
\begin{equation}\label{eq:main-recurrence}
  \SafeTarskiStar(n,k)
  \leq O\left(k2^k\log n\right)\cdot
   \SafeTarskiStar(n,\lfloor k/2\rfloor-1).
\end{equation}
This finishes the proof of the theorem.
\end{proof}

\section{The Lower Bound}
In this section, we will prove the query lower bound. 
\MainLower*

We consider a distribution $\cF$ of functions mapping from
a hyperrectangle (informally referred to as {\em cube} going forward)  %
to itself. 
We will use $Q(\cF)$ to denote the minimum average-case depth of a deterministic decision tree that always finds a fixed point for a function drawn from the distribution $\cF$. 
By Yao's minimax principle, this gives the same lower bound $Q(\cF)$ against any randomized algorithm in its worst case. 

Our construction is inductive. 
We use $\ell$ to denote the level of the induction, which gives us an instance with a query lower bound $\Omega_{\ell}((\log n)^{\ell+1})$. 
Throughout the induction, we consider a ``base side-length'' parameter $m$, which we will describe later. 
In the inductive step that lifts from level $\ell\!-\!1$ to level $\ell$, we compose an outer instance, which is a $1$-dimensional Tarski instance with side-length $m$, and an inner instance, which is a level-$(\ell\!-\!1)$ instance.
The outer instance is determined by a hidden number $t^*$ in $[m]$, and is defined as $f(t)=t+\mathsf{sign}(t^*-t)$.
It is equivalent to a binary search instance for finding a hidden number in $[m]$, where queries provide comparisons between any number and the hidden number, and thus has a query lower bound of $\Omega(\log m)$. 
This observation allows us to prove a stronger ``high-probability'' statement for our inductive hypothesis: any deterministic decision tree with depth $q_{\ell,m}=\Theta_{\ell}((\log m)^{\ell+1})$ can only find a fixed point of a function mapping from $[m^\ell]^{\Theta(4^{\ell})}$ to itself with a small probability. 

\begin{restatable}[]{theorem}{InductiveTailBound}
\label{theorem:inductive-tail-bound}
There exists an absolute constant $c_0\in (0,1/3)$ such that the following holds. 
Let $m \ge 200$ {\em(base side-length)} and $\ell \ge 1$ {\em (level)} be integers, and let $n_{\ell,m}=m^{\ell}$ {\em( side length)} and $k_{\ell}=2(4^\ell-1)/{3}$ {\em (dimension)}.
Then there exists a distribution $\cF_{\ell,m}$ of monotone functions $F:[n_{\ell,m}]^{k_{\ell}} \to [n_{\ell,m}]^{k_{\ell}}$ such that:

\begin{itemize}
    \item Each function $F\in \mathsf{supp}(\cF_{\ell,m})$ has a unique fixed point.
    \item Let $F\sim \cF_{\ell,m}$. For any deterministic decision tree with oracle access to $F$ and depth at most $q_{\ell,m}=c_0\big(\frac{\log m}{512\ell^3}\big)^{\ell+1}$, the probability that it finds the fixed point of $F$ is at most $\eps_{\ell,m}=3m^{-c_0-\frac{1}{8\ell}}$.
\end{itemize}
\end{restatable}
Note that \Cref{theorem:inductive-tail-bound} implies \Cref{theorem: lower}:
We let $\ell=\lfloor\log_4(6k+4)\rfloor-1$.
Then, $6k_\ell+4=4^{\ell+1}\leq 6k+4$, and thus $k_\ell\leq k$.
For all sufficiently large $n$, let $m=\lfloor n^{1/\ell}\rfloor$.
Padding the side length from $m^\ell$ to $n$ and adding $k-k_\ell$ coordinates that are constantly equal to $1$ together preserve the monotonicity and the uniqueness of the fixed point. Each query to the padded function can be simulated by one query to the original function. Since $\log m\geq(\log n)/(2\ell)$ and $\eps_{\ell,m}=o(1)$, \Cref{theorem:inductive-tail-bound} gives the following average-case lower bound for any deterministic algorithm. 
\[
q_{\ell,m}(1-\eps_{\ell,m})
\geq c_0\left(\frac{\log n}{1024\ell^4}\right)^{\ell+1} \cdot 0.5 = \Omega\left(\left(\frac{\log n}{\polylog(k)}\right)^{\lfloor\log_4(6k+4)\rfloor}\right).
\]
Finally, applying Yao's minimax principle as discussed above, we have proved \Cref{theorem: lower}.

When $\ell=1$ (the base case), we use the existing construction of \cite{EPRY20} that works for $2$-dimensional Tarski. %
For $\ell\geq 2$, we use an amplification step that lifts the level-$(\ell\!-\!1)$ distribution to the level-$\ell$ distribution, and increases the lower bound by an $\Omega_{\ell}(\log m)$ factor. We present the step in \Cref{subsec:lower-amp-step} and analyze it in \Cref{subsec:lower-amp-basic,sec:randomized-lower-bound}.
In the rest of this section, we use $\Delta(A)$ as the set of all distributions supported on $A$, and identify $\mathbf1$ as the all-ones vector.

\subsection{The base case: a distribution of two-dimensional monotone functions}
\label{subsec:lower-basecase}
Our construction starts from the two-dimensional herringbone construction of Etessami, Papadimitriou, Rubinstein, and Yannakakis~\cite{EPRY20}. 
The construction gives our desired bound for $\ell=1$ with $k_1=2$ dimensions.\footnote{Although \cite{EPRY20} only explicitly states an \(\Omega((\log m)^2)\) randomized query lower bound, the stronger probability bound in~\Cref{lem:lower-basecase} also follows from their construction and analysis.}

\begin{lemma}[\cite{EPRY20}]
\label{lem:lower-basecase}
There is a universal constant $c_0>0$ such that, for every integer $m\geq 200$,
there is a distribution $\cF$ of monotone functions $[m]^2\to[m]^2$, such that 
\begin{itemize}
    \item Each function $F\in \mathsf{supp}(\cF)$ has a unique fixed point.
    \item Let $F\sim \cF$. For any deterministic decision tree with oracle access to $F$ and depth at most $q=c_0(\log m)^2$, the probability that it finds the fixed point of $F$ is at most $m^{-c_0-\frac{1}{8}}$.
\end{itemize}
\end{lemma}

\subsection{The amplification step}
\label{subsec:lower-amp-step}
In this subsection, we present the amplification step of our construction. 
We fix $m$ as ``base side-length'' parameter. 
We will start with a distribution $\cH$ whose support consists of monotone functions $H:[n]^k\to [n]^k$ with a unique fixed point. 
The amplification step will produce a distribution 
\begin{align}
    \label{eqn:Amp-statement}
  \Amp_{m}(\cH)
  \in
  \Delta\left(\bigl\{f:[m]^2 \times [mn]^{4k}\longrightarrow [m]^2 \times [mn]^{4k}\bigr\}\right)
\end{align}
whose support consists of monotone functions with unique fixed points. 

In this subsection, we will use the notation of complement $\overline{z}$ for any $z\in[n]^k$: 
$$
  \overline z=(n+1)\one-z.
$$
It is clear that $\overline{(\overline{z})} = z $ for any $z\in [n]^k$. 

\subsubsection{Step 1: putting all candidate fixed points in an anti-chain}

Let $H: [n]^k \to [n]^k$ be a function in the support of $\cH$. We first define the following function $G_H:[n]^{2k} \to [n]^{2k}$:
\begin{equation}\label{eq:dual-gate}
  G_H(x,u)
  =\bigl(H(x),\overline{H(\overline u)}\bigr),
  \qquad
  \forall\ (x,u)\in[n]^{2k}.
\end{equation}
Suppose that $P_H$ is the unique fixed point of $H$. 
A key feature of the function $G_H$ is that it has a unique fixed point, as shown by the following lemma. 
\begin{lemma} 
\label{lem:unique-fixed-point-G_H}
For any $x, u\in [n]^{k}$, we have $G_H(x,u) = (x,u)$ if and only if $x=P_H$ and $u=\overline{P_H}$.
\end{lemma}
\begin{proof}
    If $G_H(x,u) = (x,u)$, we have $H(x) = x$ and $\overline{H(\overline{u})} = u$. 
    The second equality gives us $H(\overline{u}) = \overline{u}$. 
    Since $P_H$ is the unique fixed point of $H$, we have $x=P_H$ and $\overline{u} = P_H$, and thus $u=\overline{P_H}$.

    Conversely, if we have $x=P_H$ and $u=\overline{P_H}$, we have $$G_H(x,u) = (H(P_H) , \overline{H(P_H)}) = (P_H, \overline{P_H}) = (x,u)~.$$
    Hence, we have completed the proof. 
\end{proof}
\Cref{lem:unique-fixed-point-G_H} implies that the unique fixed point must lie in $$C=\{(z,\overline{z}): z\in [n]^k\} \subseteq \left\{x\in [n]^{2k}: \sum_{i\in [2k]} x_i=k(n+1)\right\}~.$$ 
Because $\{x\in [n]^{2k}: \sum_{i\in [2k]} x_i=k(n+1)\}$ is an anti-chain, after this step, we can ensure that all potential fixed points of $G_H$ are in the fixed anti-chain $C$ (that is independent of the function $H$).

\subsubsection{Step 2: Hiding one new bit in a sign function}
We partition the space $[n]^{2k}$ into two subsets: $D_C$ and $U_C$, where the set $D_C$ consists of all points that are below some point in the anti-chain $C$ and $U_C$ consists of all other points. %
$$
  D_C=\{w\in [n]^{2k}: \exists w'\in C, w\preceq w'\}, \text{\quad and \quad}
  U_C=[n]^{2k} \setminus D_C.
$$
We remark here that $C\subseteq D_C$ and that both $D_C$ and $U_C$ are independent of $H$.

Next, fixing a function $H\in \mathsf{supp}(\cH)$, we consider hiding one bit $c\in \{0,1\}$ at the unique fixed point of $G_H$. More specifically, we consider the following sign function ${\mathsf{S}_{H,c}}: [n]^{2k} \to \{0,1\}$:
\begin{equation}\label{eq:local-label}
    \forall X\in [n]^{2k}, \quad 
  {\mathsf{S}_{H,c}}(X)=
  \begin{cases}
    c,&X=(P_H, \overline{P_H}),\\
    0,&X\in D_C \setminus\{(P_H, \overline{P_H})\},\\
    1,&X\in U_C.
  \end{cases}
\end{equation}

Intuitively, this sign function is hiding the bit $c$ unless you have reached the fixed point of $H$: for all other points, the function value is automatically 0 if it is in $D_C$, while the value is 1 if it is in $U_C$. 
Therefore, to learn $c$ from querying ${\mathsf{S}_{H,c}}$, the only way is to find the fixed point $P_H$ of $H$ and query $X=(P_H, \overline{P_H})$.

Another key property of this sign function is that it is trivially monotone because $D_C$ includes all points that are below some point of $D_C$. We will formalize and use this property later in proving that our constructed distribution of functions is monotone (\Cref{lem:amplifier-monotone}).

\subsubsection{A failed attempt to put the first two steps together}
\newcommand{\Proj}{\mathcal{P}}
With the first two steps, one may be curious if we can put them together in a $2k+1$-dimensional instance, with the outer 1-dimensional Tarski instance encoded in the first coordinate and the inner instances $G_H$ embedded in the last $2k$ coordinates. 

More specifically, we consider constructing a function $F:[m]\times [mn]^{2k} \to [m] \times [mn]^{2k}$. 
We call the first coordinate $a\in [m]$ the ``sign coordinate'' and the last $2k$ coordinates $X\in [mn]^{2k}$ the ``value coordinates''. 
Accordingly, the function value is partitioned into $F(\cdot)=(F_a(\cdot), F_X(\cdot))$ with $F_a(\cdot)\in [m]$ and $F_X(\cdot) \in [mn]^{2k}$.

We consider $m$ boxes for the ``value coordinates''. The $t$-th box is defined as all points whose coordinates are between $(t-1)\cdot n+1$ and $t\cdot n$, i.e., 
\begin{align}
    \label{eqn:definition of B_t}
    B_t = \{X\in [mn]^{2k}: \forall i\in [2k], (t-1)\cdot n < X_i \leq t\cdot n\},
\end{align}
which forms a $2k$-dimensional hypercube of side-length $n$ and it is straightforward to consider a translated copy of $G_H$ (for $H\in \mathsf{supp}(\cH)$) as a function mapping from any $B_t$ to itself.
For the value coordinates, we sample $H_1,\dots,H_m\sim \cH$ from the original distribution $\cH$, and define the value $F_X(a,X)$ as the value of $\Proj_a(X)$ in $G_{H_a}$'s translated copy, where $\Proj_a(X)$ is the projection of $X$ onto box $B_a$. This is a natural way to guarantee the monotonicity for the value coordinates (because any point $W\in B_a$ is below any point $W'\in B_{a+1}$).

At the same time, we consider hiding a random number $t^*\sim \mathsf{uniform}([m])$ in the outer 1-dimensional instance.
We will use this outer instance to define the displacement $F_a(a,X)-a$ for each point. 
Ideally, we want to place the sign $\mathsf{sign}(t^*-a)$ on the  ``fixed point'' of each box (points $(a,X)$ with $X=FP_{a}$, where $FP_a$ is the fixed point for the translated copy of $G_{H_a}$ in $B_a$), and place default values independent of $t^*$ at all other points, similarly to \Cref{eq:local-label}. 
However, this requires three possible displacement values $-1, 0, 1$ instead of the two values $0,1$ used in \Cref{eq:local-label}.
Next, we argue that this could be problematic to guarantee the monotonicity of the sign coordinate. 

Let $\mathsf{MIN}_a$ and $\mathsf{MAX}_a$ be the minimum and maximum elements of $B_a$, respectively. 
Because $\mathsf{MIN}_a \preceq FP_a \preceq \mathsf{MAX}_a$, the monotonicity of $F$ requires our definition of $F_a$ to satisfy $F_a(a, \mathsf{MIN}_a)\leq F_a(a, FP_a) \leq F_a(a, \mathsf{MAX}_a)$. 
For $1<a<m$, it is possible that the displacement $F_a(a, FP_a)-a$ is $-1, 0$ or $1$, depending on our choice of $t^*$.
Therefore, the default value of $F_a(a, \mathsf{MIN}_a)$ must be at most $a-1$, while the default value of $F_a(a, \mathsf{MAX}_a)$ must be at least $a+1$.
However, this could lead to a violation of monotonicity in other pairs of points. 
Consider two points for some $t\in [2,m-2]$:
$$W=(t,\mathsf{MAX}_t)\quad \text{and}\quad W'=(t+1, \mathsf{MIN}_{t+1})~.$$ We have $W\preceq W'$. 
According to the above discussion, because these default values must reveal no information about $t^*$, we have $F_a(W) \geq t+1$ and $F_a(W') \leq (t+1)-1=t$, which violate $F(W)\preceq F(W')$.

\subsubsection{Step 3: Defining the new distribution of monotone functions}

In this step, we define the new distribution of monotone functions. 
Beyond our failed attempt discussed earlier, we consider $4k+2$-dimensional instances, which can be regarded as doubling the failed construction. 
The reason we can circumvent the monotonicity issue in the failed attempt is that we can use displacement values in $\{0,1\}$ on one dimension and use $\{-1,0\}$ on another. 
Because we don't have displacement $-1, 1$ on the same dimension, this helps to rule out the violation we discussed earlier. 
To prevent introducing multiple fixed points, we then swap the updated values of the two sign coordinates.

In this construction, we will call the first two coordinates ``sign coordinates'' and the last $4k$ coordinates ``value coordinates''. 
We divide the ``value coordinate'' into two groups: the first consists of the first $2k$ value coordinates, and the second consists of the last $2k$ value coordinates. 
In addition, we consider $m$ boxes that are restricted to each group. 
We use the same definition of $B_t$ as in \Cref{eqn:definition of B_t}.
Formally, we define $\Proj_t(X)$ as the projection of $X\in [mn]^{2k}$ onto the box $B_t$ as follows:
\begin{align*}
    [\Proj_t(X)]_i = \min\{t\cdot n, \max\{(t-1)\cdot n + 1, X_i\}\}~.
\end{align*}

\begin{figure}[t]
    \centering
    \tikzset{%
    every neuron/.style={
        circle,
        draw,
        minimum size=10pt,
        inner sep=0pt,
        very thin
    },
    neuron/.style={
        circle,
        fill=Purple,
        inner sep=0pt,
        minimum size=3.5pt
    },
}
\begin{tikzpicture}[x=0.9cm, y=0.9cm, >=latex, every text node part/.style={align=center}]
    \foreach \x in {1,...,27} {
      \fill[yellow!40] (5*\x-4.9 pt, 5*\x-4.9pt) rectangle (5*\x-0.1 pt, 5*\x-0.1pt);
    }
    \foreach \x in {1,...,26} {
      \path[very thin, black!10] (5*\x pt, 0pt) edge (5*\x pt, 135pt);
    }
    \foreach \x in {1,...,26} {
      \path[very thin, black!10] (0pt, 5*\x pt) edge (135pt, 5*\x pt);
    }
    \draw[thick] (0pt, 0pt) rectangle (135pt, 135pt);
    \node[] at (-10pt, 75pt) {$\mathbf{b}$};
    \node[] at (75pt, -10pt) {$\mathbf{a}$};
    \node[] at (-5pt, 132.5pt) {\scriptsize $m$};
    \node[] at (-5pt, 2.5pt) {\scriptsize $1$};
    \node[] at (132.5pt, -5pt) {\scriptsize $m$};
    \node[] at (2.5pt, -5pt) {\scriptsize $1$};
    \node[] at (48pt, -5pt) {\scriptsize $t^*$};
    \node[] at (-5pt, 48pt) {\scriptsize $t^*$};
    \path[dashed] (45pt, 0pt) edge (45pt, 135pt);
    \path[dashed] (0pt, 50pt) edge (135pt, 50pt);
    \path[->,thick,Red] (10pt, 22.5pt) edge (35pt, 22.5pt);
    \path[->,thick,Blue] (92.5pt, 105pt) edge (92.5pt, 80pt);
    \path[->,thick,Purple] (10pt, 105pt) edge (35pt, 80pt);
    \node[] at (92.5pt, 22.5pt) {\Large \color{Green} $\mathbf{0}$};
\end{tikzpicture}
    \caption{\small Illustration of updates on the sign coordinates $a,b$ before swapping them, when the translated projections of $X$ and $Y$ are fixed points of $G_{H_a}$ and $G_{H_b}$. We increase $a$ by one when $a<t^*$ and decrease $b$ by one when $b>t^*$. Hence, the bottom-left region moves right, the top-right region moves down, the top-left region moves both right and down, and the bottom-right region remains unchanged. The yellow cells mark the diagonal $a=b$, which is fixed under our swapping operation. This illustrates how the update and swap rules single out $(t^*,t^*)$ as the sign coordinates of the unique fixed point of our constructed functions.}
    \label{fig:sign-coordinates}
\end{figure}
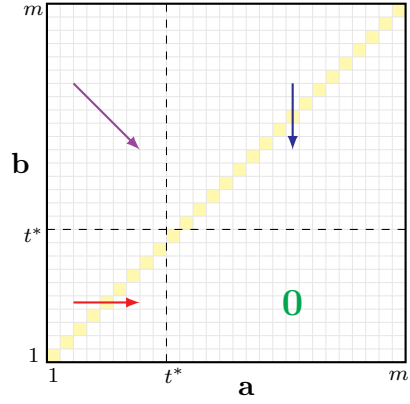

Next, we construct a random function $F:[m]^2 \times [mn]^{4k} \to [m]^2 \times [mn]^{4k}$. The distribution of the function is then our desired $\Amp_m(\cH)$.
\begin{enumerate}
    \item We first sample $m$ i.i.d.~functions $H_1, H_2, \dots, H_{m}$ from the original distribution $\cH$. 
    For simplicity, we write $P_i$ for the unique fixed point of $H_i$. 
    \item We sample a hidden number $t^*$ from $[m]$ uniformly at random. We define the sign $c_t = \mathbf{1}(t\leq t^*)$ for each $t\in [m]$. We also define $c_{m+1}=0$ for simplicity. 
    \item For each $W\in [m]^2 \times [mn]^{4k}$, we rewrite $W=(a,b,X,Y)$ where $a,b\in [m]$ and $X,Y\in [mn]^{2k}$. 
    \item For each $W\in [m]^2 \times [mn]^{4k}$, we rewrite $F(W)=(F_{a}(W), F_b(W), F_X(W), F_Y(W))$ where $F_a(W),F_b(W)\in [m]$ and $F_X(W),F_Y(W)\in [mn]^{2k}$. 
    \item \textbf{Defining the function value for the ``value coordinates'':} we pair $a$ with $X$ and $b$ with $Y$. The value of $F_X$ is determined by the pair $(a,X)$, while the value of $F_Y$ is determined by the pair $(b,Y)$. The formulations for $F_X$ and $F_Y$ are identical based on the pair. For $F_X$, we first compute the projection of $X$ onto $B_a$, translate it to the space $[n]^{2k}$, and then consider its value of $G_{H_a}$ (translated back to the space $B_a$). 
    More specifically, we define
    \begin{align*}
        F_X(W) &= (a-1)\cdot n \cdot \mathbf{1} + G_{H_a}(\Proj_a(X)-(a-1)\cdot n \cdot \mathbf{1})~,
        \\
        F_Y(W) &= (b-1) \cdot n \cdot \mathbf{1} + G_{H_b}(\Proj_b(Y)- (b-1)\cdot n\cdot \mathbf{1})~.
    \end{align*}
    \item \textbf{Defining the function value for the ``sign coordinates'':} We will hide the sign $c_t$ (comparing $t$ with $t^*$) in these coordinates. As in the definition for the ``value coordinates'', we still pair $(a,X)$ and $(b,Y)$. We consider slightly updating the ``sign coordinates'' (also demonstrated in \Cref{fig:sign-coordinates}) and then swapping them to define their function value, following the rules below:
    \begin{itemize}
        \item We move $a$ based on $X$'s sign in the sign function ${\mathsf{S}_{H_a,c_{a+1}}}$ and place its value on $F_b$:
        \begin{align*}
            F_b(W)&=\min\{m,a+{\mathsf{S}_{H_a,c_{a+1}}}(\Proj_a(X)-(a-1)\cdot n\cdot \mathbf{1})\}.
        \end{align*}
        Let $X'=\Proj_a(X)-(a-1)n\cdot \mathbf{1}$ be translated projection of $X$. 
        This gives $F_b(W)=a$ if $X'\in D_C\setminus \{(P_{a}, \overline{P_{a}})\}$ or $F_b(W)=\min\{m,a+1\}$ if $X'\in U_C$. 
        More importantly, when $X'$ equals the fixed point $(P_a, \overline{P_a})$, this gives $F_b(W)=a$ when $a\geq t^*$ but gives $F_b(W)=a+1$ when $a<t^*$.
        \item We move $b$ based on $Y$'s sign in the sign function ${\mathsf{S}_{H_b,c_b}}$ and place its value on $F_a$:
        \begin{align*}
            F_a(W)&=\max\{1,b-1+{\mathsf{S}_{H_b,c_b}}(\Proj_b(Y)-(b-1)\cdot n\cdot \mathbf{1})\}.
        \end{align*}
        Let $Y'=\Proj_b(Y)-(b-1)n\cdot \mathbf{1}$ be the translated projection of $Y$. 
        This gives $F_a(W)=\max\{1,b-1\}$ if $Y'\in D_C\setminus \{(P_{b}, \overline{P_{b}})\}$ or $F_a(W)=b$ if $Y'\in U_C$. 
        More importantly, when $Y'$ equals the fixed point $(P_b, \overline{P_b})$, this gives $F_a(W)=b$ when $b\leq t^*$ but gives $F_a(W)=b-1$ when $b>t^*$.
    \end{itemize}
    Overall, we conclude the following intuition for the definition on the ``sign coordinates'':
    \begin{itemize}
        \item We want $F_a$ (or, $F_b$) to reveal no information about $t^*$ if the translated projection of $Y$ (or, $X$) is not a fixed point. 
        \item Assume that the translated projections of $X,Y$ are both fixed points. The swapping idea ensures that $a,b$ should be equal to each other in any fixed point $W$ of $F$. This ensures that the fixed point is unique (and thus finding it requires finding the hidden $t^*$). Then, we only move exactly one of $a$ and $b$ in one (fixed) direction when $a=b\neq t^*$, in order to guarantee monotonicity. 
    \end{itemize}
\end{enumerate}

\subsection{Basic properties of the constructed distribution of functions}
\label{subsec:lower-amp-basic}
In this subsection, we show that any function in the support of the constructed distribution $\Amp_{m}(\cH)$ is monotone and has a unique fixed point. 
\subsubsection{Uniqueness of the fixed point}

\begin{lemma}\label{lem:amplifier-unique}
Consider any $F\in \mathsf{supp}(\Amp_{m}(\cH))$.
Suppose that the sampled $m$ functions from $\cH$ are $H_1, H_2, \dots, H_{m}$, and their unique fixed points are respectively $P_1, P_2, \cdots, P_{m} \in [n]^{k}$. 
Suppose that $t^*$ is the sampled hidden number in $[m]$.
Let $$P^* = (P_{t^*}+(t^*-1)n\cdot \mathbf{1}, \overline{P_{t^*}}+(t^*-1)n\cdot \mathbf{1}).$$
Then, $F$ has the unique fixed point
  $(t^*,t^*,P^*, P^*)$.
\end{lemma}

\begin{proof}
First, we show that $W^*=(t^{*},t^{*},P^{*},P^{*})$ is a fixed point by simple calculations.
Since $P_{t^*}\in [n]^k$, we have $P^*\in B_{t^*}$ by the definition of $B_{t^*}$ and thus $\Proj_{t^*}(P^*)=P^*$.
Hence, we have $\Proj_{t^*}(P^*) - (t^*-1)n\cdot \mathbf{1}=(P_{t^*},\overline{P_{t^*}})$. 
Because we have $G_{H_{t^*}}(P_{t^*}, \overline{P_{t^*}})=(P_{t^*}, \overline{P_{t^*}})$, we have $F_X(W^*)=P^*$ and $F_Y(W^*)=P^*$ by their definitions.
By definitions of $F_a$ and $F_b$, we have 
\begin{align*}
    F_b(W^*)& = \min\{m, t^* + c_{t^*+1}\} = t^*~, \text{ }
    and  \tag{$c_{t^*+1}=0$ as $t^*+1>t^*$}\\
    F_a(W^*) &= \max\{1, t^* - 1 + c_{t^*}\} = t^*~. \tag{$c_{t^*}=1$ as $t^*\leq t^*$}
\end{align*}
Therefore, we have $F(W^*)=W^*$.
Next, we show that this is the only fixed point of $F$.
Let $W=(a,b,X,Y)$ be any fixed point of $F$. 

First, we prove the characterization of $X,Y$: $X=(P_a, \overline{P_a}) + (a-1)n\cdot \mathbf1$ and $Y=(P_b, \overline{P_b}) + (b-1)n\cdot \mathbf1$.
Since $X=F_X(W)\in B_a$ and $Y=F_Y(W)\in B_b$, we have
$X\in B_a$ and $Y\in B_b$ and thus $\Proj_a(X)=X$ and $\Proj_b(Y)=Y$. 
According to the definitions of $F_X,F_Y$, we have 
\begin{align*}
X=F_X(W)&=(a-1)n\cdot \mathbf{1} + G_{H_a}\bigl(X-(a-1)n\cdot \mathbf1\bigr)\\
Y=F_Y(W)&=(b-1)n\cdot \mathbf{1} + G_{H_b}\bigl(Y-(b-1)n\cdot \mathbf1\bigr).
\end{align*}
This implies that $X-(a-1)n\cdot \mathbf1$ is a fixed point of $G_{H_a}$ and $Y-(b-1)n\cdot \mathbf1$ is a fixed point of $G_{H_b}$. By \Cref{lem:unique-fixed-point-G_H}, we have obtained the characterization of $X,Y$.

Second, we show that $a=b=t^*$. 
Given the characterization of $X$ and $Y$, we have
\begin{align*}
    \Proj_a(X) - (a-1)n\cdot \mathbf1 =  (P_a, \overline{P_a}), \quad \text{and}\quad \Proj_b(Y) - (b-1)n\cdot \mathbf1 = (P_b, \overline{P_b}).
\end{align*}
According to the definition of $F_a, F_b$, we have 
\begin{equation}
a=F_a(W) = \max\{1,b-1+c_b\},\quad\text{and}\quad
b=F_b(W) = \min\{m,a+c_{a+1}\}.
\label{eq:lemma8-sign-coordinates}
\end{equation}
Suppose that $a<t^*$. Then, we have $c_{a+1}=1$  and thus $b=a+1\leq t^*$. Therefore, $a = \max\{1,b\}=b$, contradicting $b=a+1$.
Suppose that $a>t^*$. Then, we have $c_{a+1}=0$ and thus $b=a>t^*$. Therefore, $a=b-1$, contradicting $b=a$. 
Therefore, we have $a=t^*$. Then, $c_{a+1}=0$ and we get $b=a$.

In conclusion, we have $a=b=t^*$ and $X=Y=(P_{t^*},\overline{P_{t^*}}) + (t^*-1)n\cdot \mathbf1 = P^*$. This gives the uniqueness of the fixed point of $F$.
\end{proof}

\subsubsection{Monotonicity}

\begin{lemma}\label{lem:amplifier-monotone}
Every member $F\in \mathsf{supp}(\Amp_{m}(\cH))$ is monotone.
\end{lemma}
\begin{proof}
First, we prove that $G_H$ is monotone for any $H\in \mathsf{supp}(\cH)$. Suppose that $(x,u)\preceq(x',u')$. Because of the monotonicity of $H$, we have $H(x)\preceq H(x')$. 
Also, $u\preceq u'$ implies that $\bar u\succeq\bar u'$, and hence $H(\bar u)\succeq H(\bar u')$. 
Taking complements gives $\overline{H(\bar u)}\preceq\overline{H(\bar u')}$.
By definition of $G_H$, we have $G_H(x,u)\preceq G_H(x',u')$.

Next, we prove that ${\mathsf{S}_{H,c}}$ is monotone for any $H\in \mathsf{supp}(\cH)$ and $c\in \{0,1\}$. 
Let $Z\preceq Z' \in [n]^{2k}$.
If ${\mathsf{S}_{H,c}}(Z)=0$ or $Z=Z'$, it is clear that ${\mathsf{S}_{H,c}}(Z) \leq {\mathsf{S}_{H,c}}(Z')$.
Otherwise, we have either $Z\in U_C$, or $Z=(P_H,\overline{P_H})$ and $c=1$.
In the former case, if $Z'\in D_C$, we can find $Z''\in C$ such that $Z'\preceq Z''$ and thus $Z\preceq Z''$. This implies $Z\in D_C$ by the definition of $D_C$, which contradicts that $Z\in U_C$ and $U_C, D_C$ are disjoint. 
In the latter case, $Z'\in U_C$ as well; otherwise, $Z\prec Z'\preceq Z''$ for some $Z''\in C$, contradicting that $Z\in C$ and $C$ forms an anti-chain. Thus ${\mathsf{S}_{H,c}}(Z')=1={\mathsf{S}_{H,c}}(Z)$ in either case. 

We now use these two properties to prove that $F$ is monotone.
It suffices to compare two points $W\preceq W'$ that differ in exactly
one coordinate, because of the transitivity of $\preceq$. 
We write $W=(a,b,X,Y)$ and discuss the following cases.

\smallskip\noindent
\textbf{Only $a$ increases.}
Write $W'=(a',b,X,Y)$ with $a'>a$. The outputs $F_a$ and $F_Y$
are unchanged. Since every sign function takes values in $\{0,1\}$,
\[
F_b(W)\le \min\{m,a+1\}\le a'\le F_b(W').
\]
Moreover, because $F_X(W)\in B_a$ and $F_X(W')\in B_{a'}$, for every
$i\in[2k]$,
\[
[F_X(W)]_i\le an<(a'-1)n+1\le [F_X(W')]_i.
\]
Therefore, $F(W)\preceq F(W')$.

\smallskip\noindent
\textbf{Only $b$ increases.}
Write $W'=(a,b',X,Y)$ with $b'>b$. The outputs $F_b$ and $F_X$
are unchanged, while
\[
F_a(W)\le b\le b'-1\le F_a(W').
\]
Moreover, because $F_Y(W)\in B_b$ and $F_Y(W')\in B_{b'}$, for every $i\in [2k]$,
\[
[F_Y(W)]_i\le bn<(b'-1)n+1\le [F_Y(W')]_i.
\]
Therefore,  $F(W)\preceq F(W')$.

\smallskip\noindent
\textbf{Only one coordinate of $X$ increases.}
Write $W'=(a,b,X',Y)$, where $X\prec X'$. 
The outputs $F_a$ and $F_Y$ are unchanged. Since $\Proj_a$ and $G_{H_a}$ are monotone, we  have $F_X(W)\preceq F_X(W')$ by the definition of $F_X$.
Moreover, because of the monotonicity of ${\mathsf{S}_{H_a,c_{a+1}}}$, we have
\begin{align*}
F_b(W)&=\min\{m,a+{\mathsf{S}_{H_a,c_{a+1}}}(\Proj_a(X)- (a-1)n\cdot \mathbf{1})\}\\
&\le \min\{m,a+{\mathsf{S}_{H_a,c_{a+1}}}(\Proj_a(X')-(a-1)n\cdot \mathbf{1})\}=F_b(W').
\end{align*}

\smallskip\noindent
\textbf{Only one coordinate of $Y$ increases.}
Write $W'=(a,b,X,Y')$, where $Y\prec Y'$. The outputs $F_b$
and $F_X$ are unchanged. Since $P_b$ and $G_{H_b}$ are monotone, we have
$F_Y(W)\preceq F_Y(W')$ by the definition of $F_Y$.
Moreover, because of the monotonicity of ${\mathsf{S}_{H_b,c_b}}$, we have
\begin{align*}
F_a(W)&=\max\{1,b-1+{\mathsf{S}_{H_b,c_b}}(\Proj_b(Y)-(b-1)n\cdot \mathbf{1})\}\\
&\le \max\{1,b-1+{\mathsf{S}_{H_b,c_b}}(\Proj_b(Y')-(b-1)n\cdot \mathbf{1})\}=F_a(W').
\end{align*}
Therefore, we have $F(W)\preceq F(W')$ in every case, proving that $F$ is monotone.
\end{proof}

\subsection{Query lower bound for the constructed distribution of functions}
\label{sec:randomized-lower-bound}

In this subsection, we analyze the amplification step and prove \Cref{theorem:inductive-tail-bound}, restated below. 
\InductiveTailBound*
Suppose that $\cH$ is the level-$(\ell-1)$ distribution and satisfies the level-$(\ell-1)$ guarantee in \Cref{theorem:inductive-tail-bound}. We will show that a padded version of our amplification $\cF=\Amp_{m}(\cH)$ satisfies the level-$\ell$ guarantee.
To make the notation simpler (hiding the subscript), we say functions in $\mathsf{supp}(\cH)$ map from $[n]^k$ to itself, and the induction hypothesis is 
\vspace{-0.3em}
\begin{quote}
{\bf Induction hypothesis:}
{\it
    the probability that any deterministic decision tree with depth at most $q$ finds the fixed point of $H\sim \cH$ is at most $\eps$.
}
\end{quote}
\vspace{-0.3em}
Recall our notations in \Cref{subsec:lower-amp-step}. $H_1, \cdots, H_m$ are i.i.d.~samples from $\cH$ and $t^*\sim \mathsf{uniform}([m])$. $P_1,\cdots,P_m$ are the unique fixed points of $H_1,\cdots,H_m$ respectively.
We consider the signs $c_{t} = \mathbf{1}(t\leq t^*)$ and $c_{m+1}=0$.
Then, we summarize the definition $F(W)=(F_a(W), F_b(W), F_X(W), F_Y(W))$ for $W=(a,b,X,Y)$ in terms of $H_1, \dots, H_m$ and $c_1,\dots, c_{m+1}$ as follows.
\begin{itemize}
    \item $F_X(W)$ uses $a,X$. Let $\Proj_a(X)-(a-1)n\cdot \mathbf{1}=(x,u)$. We make two queries $H_a(x)$ and $H_a(\overline{u})$. The values are then determined by these queries. 
    \item $F_Y(W)$ uses $b,Y$. Let $\Proj_b(Y)-(b-1)n\cdot \mathbf{1}=(y,v)$. We make two queries $H_b(y)$ and $H_b(\overline{v})$. The values are then determined by these queries.
    \item $F_a(W)$ uses $b,Y$. The value is then determined by $c_b$ only  when $\Proj_b(Y) - (b-1)n\cdot \mathbf1$ is the unique fixed point of $G_{H_b}$ (whether this holds can be determined by the values of $F_Y(W)$). Otherwise, it has a default value (based on $Y$).
    \item $F_b(W)$ uses $a,X$. The value is then determined by $c_{a+1}$ only when $\Proj_a(X) - (a-1)n\cdot \mathbf1$ is the unique fixed point of $G_{H_a}$ (whether this holds can be determined by the values of $F_X(W)$). Otherwise, it has a default value (based on $X$).
\end{itemize}
Next, we present the key lemma in our analysis that establishes a similar probability bound for $\Amp_{m}(\calH)$.

\begin{lemma}
\label{lem:randomized-one-step}
Assume the induction hypothesis for $\cH$. For every deterministic decision tree of depth at most $q'=\frac{q\log m}{64}$, the probability that it finds the fixed point of $F\sim \Amp_{m}(\cH)$ is at most
$\eps'=\frac{1}{\sqrt m}+\frac{q\eps}{3}$.
\end{lemma}

\begin{proof}
Consider any deterministic decision tree of depth at most $q'$. 
For a query $W=(a,b,X,Y)$, using our previous discussions, we can make at most $4$ queries in total to $H_1,\ldots,H_m$ to determine $F_X(W)$ and $F_Y(W)$. 
Further, the value of $F_b(W)$ (or $F_a(W)$) is independent of any sign $c_t$ unless $X$'s projection $\Proj_a(X)-(a-1)n\cdot \mathbf1$ (resp., for $\Proj_b(Y)-(b-1)n\cdot \mathbf1$) equals $(P_a,\overline{P_a})$ (resp., for $(P_b,\overline{P_b})$), where $P_a$ (resp., $P_b$) is the unique fixed point of $H_a$ (resp., $H_b$). 
This means that a query to $F(W)$ only reveals at most four queries in total, to $H_1,\cdots,H_m$; and the number of different signs $c_t$ revealed by $F(W)$ is upper bounded by the number of queries with different $(a,X)$ such that $\Proj_a(X)-(a-1)n\cdot \mathbf1=(P_a,\overline{P_a})$ plus the number of queries with different $(b,Y)$ such that $\Proj_b(Y)-(b-1)n\cdot \mathbf1=(P_b,\overline{P_b})$.

On the other hand, if the decision tree successfully finds the unique fixed point of $F$, it must learn $t^*$ because of \Cref{lem:amplifier-unique}. Therefore, we can upper bound the probability that the decision tree finds the fixed point of $F$ by the probability that it correctly finds $t^*$.

Let $R$ denote the random variable representing the total number of different signs $c_t$ revealed by the queries to $F$.

First, we upper bound the probability of learning $t^*$ if the queries reveal at most $0.5\log m$ signs $c_t\in \{0,1\}$. 
We observe that because $H_1,\cdots,H_m$ are sampled independently of $t^*$, the values of the revealed queries to them do not reveal any information about $t^*$.
Therefore, conditioning on $H_1,\dots,H_m$, because any deterministic decision tree with at most $d$ 0-1 queries can have at most $2^d$ different outputs, we have 
\begin{equation*}
 \Pr\left[
 \begin{array}{c}
 \text{the decision tree finds }t^*\\
 \text{ and } R\leq 0.5\log m
 \end{array}
 \middle| H_1,\dots, H_m
 \right]
 \leq \frac{1}{m} \cdot 2^{0.5\log m} = \frac{1}{\sqrt{m}}.
 \label{eq:few-signs}
\end{equation*}
Taking the average over all possibilities of $H_1,\cdots, H_m$, we can show that the probability that the decision tree finds $t^*$ and $R\leq 0.5\log m$ is at most $1/\sqrt{m}$.

Second, we upper bound the probability that the decision tree with depth $q'$ reveals $>0.5\log m$ signs, i.e., $\Pr[R>0.5\log m]$.
Recall that we can obtain each $F(W)$ by making 4 queries, each to one of $H_1,\dots, H_m$. 
Hence, we can transform each query to $F(W)$ into 4 such queries. 
Let $\cQ$ denote the ordered list of transformed queries. 
For each $t\in[m]$, we let $E_t$ be the event that we query $P_t$ (the unique fixed point of $H_t$) to $H_t$ in $\cQ$ and we let $N_t$ denote the total number of queries to $H_t$ in $\cQ$. 
Note that the sign $c_t$ is revealed only when we query with $(a,\Proj_a(X))=\big(t-1,\ (P_{t-1},\overline{P_{t-1}})+(t-2)n\cdot \mathbf{1}\big)$ (if $t>1$) or $(b,\Proj_b(Y))=\big(t,\ (P_t,\overline{P}_t)+(t-1)n\cdot \mathbf{1}\big)$. In these cases, we query $P_{t-1}$ to $H_{t-1}$ or query $P_t$ to $H_t$ in $\cQ$.
When $c_t$ is revealed, either $E_{t-1}$ (if $t>1$) or $E_{t}$ occurs. Therefore, we have $R\leq 2\sum_{t\in [m]} \mathbf{1}(E_t)$ and thus 
\begin{align*}
    \Pr[R>0.5\log m]\leq \Pr\left[\sum_{t\in [m]}\mathbf{1}(E_t)>\frac{1}{4}\log m\right].
\end{align*}

Because the depth of the decision tree is at most $q'$, the number of queries in $\cQ$ is at most $4q'$, which means that $\sum_{t\in [m]} N_t\leq 4q'$. The number of $t\in[m]$ such that $N_t\geq q$ is at most $4q'/q$.
Also, note that when $N_t=0$, $E_t$ can't happen.
For $t$ such that $N_t\geq 1$, we let $E'_t$ denote the event that there is a query asking for $H_t(P_t)$ among the first $\leq q-1$ queries to $H_t$. 
Hence, we have 
\begin{align*}
    \Pr\left[\sum_{t\in [m]}\mathbf{1}(E_t)>\frac{1}{4}\log m\right] &\leq \Pr\left[\sum_{t\in [m]:1\leq N_t<q}\mathbf{1}(E_t)>\frac{1}{4}\log m - \frac{4q'}{q}\right]
    \\
    &\leq \Pr\left[\sum_{t\in [m]:N_t\geq 1}\mathbf{1}(E'_t)>\frac{1}{4}\log m - \frac{4q'}{q}\right]~.
\end{align*}

Consider any $t\in [m]$ such that $N_t\geq 1$.
Because $H_1,\dots, H_m$ (except $H_t$) and $t^*$ are sampled independently of $H_t$, we can simulate the following (randomized) decision tree that makes at most $q-1$ queries to $H_t$. Every time a transformed query is about $H_t$, we query it until we have asked $q-1$ queries in total; otherwise, we consider it as some random bits, taken over the randomness of $H_1,\dots, H_m$ (except $H_t$) and $t^*$; finally, we check if we have queried the fixed point of $H_t$ by the queries: if we have asked a point $x$ such that $H_t(x)=x$, we output the point; otherwise, we output arbitrarily.
According to our induction hypothesis, any decision tree with depth at most $q$ can only find the fixed point of $H_t$ with probability at most $\eps$. 
Note that this simulated decision tree finds the fixed point when $E'_t$ happens. Therefore, we have 
\begin{align*}
    \mathbb{E}\left[\sum_{t\in [m]: N_t\geq 1}\mathbf{1}(E'_t)\right] = \sum_{t\in [m]} \Pr[E'_t\text{ and }N_t\geq 1] \leq \eps \sum_{t\in [m]} \mathbb{E}[N_t] \leq 4q'\eps~.
\end{align*}
Using Markov's inequality, we have 
\begin{align*}
    \Pr\left[\sum_{t\in [m]:N_t\geq 1}\mathbf{1}(E'_t)>\frac{1}{4}\log m - \frac{4q'}{q}\right] &\leq \frac{4q'\eps}{\frac{1}{4}\log m - \frac{4q'}{q}}
    \\
    &= \frac{(q\eps\log m)/16}{\frac{1}{4}\log m - \frac{1}{16}\log m} = \frac{q\eps}{3}~.
\end{align*}
Hence, we have shown that $\Pr[R>0.5\log m]\leq \frac{q\eps}{3}$.

Finally, putting everything together, we have the following upper bound for the probability that the decision tree finds $t^*$:
\begin{align*}
    &\Pr\left[\text{the decision tree finds $t^*$}\right] \leq
    \\
    &\qquad\qquad \Pr[R>0.5\log m] + \Pr\left[\text{the decision tree finds $t^*$}\text{ and } R\leq 0.5\log m\right] \leq \frac{1}{\sqrt{m}} + \frac{q\eps}{3}~.
\end{align*}
Because the probability that the decision tree finds the fixed point of $F$ is upper bounded by the above probability, we have completed the proof. 
\end{proof}

\begin{proof}[Proof of \Cref{theorem:inductive-tail-bound}]

    Fix any $m\geq 200$. 
    The base case ($\ell=1$) follows from \Cref{lem:lower-basecase}.

Let $\ell\geq2$ and suppose that the theorem holds at level $\ell-1$. Let
$\cH=\cF_{\ell-1,m}$ and apply the amplification step to obtain $\widehat{\cF}=\Amp_m(\cH)$. Because we have $n_{\ell,m}=m n_{\ell-1,m}$ and $k_\ell=4k_{\ell-1}+2$ by their definitions, every $\widehat F\in\mathsf{supp}(\widehat{\cF})$ maps $[m]^2\times[n_{\ell,m}]^{k_\ell-2}$ to itself.
By \Cref{lem:amplifier-monotone,lem:amplifier-unique}, every such $\widehat F$ is monotone and has a unique fixed point.

To obtain $\cF_{\ell,m}$, we pad the first two coordinates that have side length $m$. For each $W\in[n_{\ell,m}]^{k_\ell}$, we let $F(W)=\widehat F(\Pi(W))$, where we sample $\widehat F \sim \widehat \cF$ and define
\[
 \Pi(W)=\bigl(\min\{W_1,m\},\min\{W_2,m\},
 W_3,\ldots,W_{k_\ell}\bigr)
\]
The function $F$ is monotone because both $\Pi$ and $\widehat F$ are
monotone. Moreover, $F$ and $\widehat F$ have the same unique fixed point because any point $W\notin [m]^2 \times [n_{\ell,m}]^{k_{\ell}-2}$ maps to $F(W)\in [m]^2 \times [n_{\ell,m}]^{k_{\ell}-2}$. 

It remains to prove the probability bound (the second bullet of the theorem). Let $q=q_{\ell-1,m}$ and $\eps=\eps_{\ell-1,m}$.
Using the definition of $q_{\ell,m}$, we have
\begin{align*}
 \frac{q_{\ell,m}}{(q_{\ell-1,m}\log m)/64}
 &=\frac{1}{8\ell^3}
   \left(\frac{\ell-1}{\ell}\right)^{3\ell}
 \leq1.
\end{align*}
Due to \Cref{lem:randomized-one-step}, any deterministic decision tree of depth at most $q_{\ell,m}$ can find the fixed point of $\widehat{F} \sim \widehat{\cF}$ with probability at most 
\begin{align*}
    \frac{2}{\sqrt m} +\frac{q_{\ell-1,m}\eps_{\ell-1,m}}3 &= \frac{2}{\sqrt{m}} + c_0\left(\frac{\log m}{512(\ell-1)^3}\right)^{\ell} \cdot m^{-c_0-\frac{1}{8(\ell-1)}}
    \\
    &\leq 2m^{-c_0-\frac{1}{8\ell}} + c_0\left(\frac{\log m}{512(\ell-1)^3}\right)^{\ell} m^{-\frac{1}{8\ell(\ell-1)}} \cdot m^{-c_0-\frac{1}{8\ell}}~.
\end{align*}
Because we can show that $\big(\frac{\log m}{512(\ell-1)^3}\big)^{\ell} < m^{1/8\ell(\ell-1)}$ by the following calculation:
\begin{align*}
\log\left(\left(\frac{\log m}{512(\ell-1)^3}\right)^\ell\right)
&\leq \ell \cdot \frac{\log m}{512(\ell-1)^3} \tag{$\log x\leq x$ for $x>0$}\\
&=\frac{\ell}{512(\ell-1)^3} \cdot \log m<\frac{1}{8\ell(\ell-1)} \cdot \log m~, \tag{$\ell\geq 2$}
\end{align*}
any deterministic decision tree of depth at most $q_{\ell,m}$ can find the fixed point of $\widehat{F} \sim \widehat{\cF}$ with probability at most $3m^{-c_0-\frac{1}{8\ell}}=\eps_{\ell,m}$. Note that any decision tree that finds a fixed point of $F$ can be simulated by a decision tree that finds a fixed point of $\widehat{F}$ with the same number of queries, and $F$ and $\widehat{F}$ have the same unique fixed point. The same probability bound also holds for $F\sim \calF_{\ell,m}$.
\end{proof}

\section*{Acknowledgment and AI disclosure}
Xi Chen and Yuhao Li were supported by NSF grants CCF-2106429 and CCF-2107187, and a grant from the AI Security Institute. 
Ruiquan Gao and Aviad Rubinstein were supported by a David and Lucile Packard Fellowship. Mihalis Yannakakis was supported by NSF grants CCF-2107187 and CCF-2212233.

The authors found both results via extensive discussions with GPT-5.6 Sol. 
The authors wrote the paper with the assistance of GPT-6 Astra and take full responsibility for all the content. 

\bibliographystyle{alpha}
\bibliography{ref}

\appendix
\section{\texorpdfstring{Proof of \Cref{theorem: upper}}{Proof of the upper bound}}\label[appendix]{app:upper-bound}

We prove \Cref{theorem: upper} by solving the recurrence in
\Cref{thm:star-recurrence} and applying
\Cref{thm:safe-tarski-equivalence,lem:tarski-star-reduction}.
The same argument allows the factor $5^k$ to be replaced by $C^k$ for any
fixed $C>4$; for simplicity, we take $C=5$.

\begin{proof}[Proof of \Cref{theorem: upper}]
Set $e(r)=\lceil\log((r+3)/3)\rceil$ for integers $r\geq0$.
We claim that there is a constant $A$ such that
$$
  \SafeTarskiStar(n,r)\leq A5^r(\log n)^{e(r)}.
$$
For $r\geq2$, write $m=\lfloor r/2\rfloor-1$. Then
$1+e(m)=e(r)$: both sides are at most an integer $j\geq1$ exactly when
$r\leq3\cdot2^j-3$. Since $e(0)=0$ and $e(1)=1$, the base cases
$\SafeTarskiStar(n,0)=1$ and $\SafeTarskiStar(n,1)=O(\log n)$, together with
\eqref{eq:main-recurrence}, give
$\SafeTarskiStar(n,r)=O((\log n)^{e(r)})$ for each fixed $r$.

Choose an absolute constant $D>0$ such that
\eqref{eq:main-recurrence} holds with multiplier $Dr2^r\log n$.
Then
$$
  Dr\,2^r5^{m-r}
  \leq\frac{Dr}{5}\left(\frac{2}{\sqrt5}\right)^r
  \longrightarrow0.
$$
Thus $Dr\,2^r5^m\leq5^r$ for all sufficiently large $r$.
Choose $A$ to cover the finitely many smaller dimensions. Induction using
\eqref{eq:main-recurrence} now gives
$$
  \SafeTarskiStar(n,r)
  \leq Dr\,2^r\log n\cdot A5^m(\log n)^{e(m)}
  \leq A5^r(\log n)^{e(r)},
$$
proving the claim.

For $k\geq3$, \Cref{lem:tarski-star-reduction} gives
\begin{equation}\label{eq:tarski-star-reduction}
  \Tarski(n,k)=O\!\left((\log n)\TarskiStar(n,k-1)\right).
\end{equation}
Combining \eqref{eq:tarski-star-reduction} with
\Cref{thm:safe-tarski-equivalence} and the claim gives
$$
  \Tarski(n,k)
  =O\!\left(5^{k-1}(\log n)^{1+e(k-1)}\right)
  =O\!\left(5^k(\log n)^{1+\lceil\log((k+2)/3)\rceil}\right).
$$
Finally, let $j=\lceil\log k\rceil\geq2$. Since
$k+2\leq2^j+2\leq3\cdot2^{j-1}$, we have
$$
  1+\left\lceil\log\frac{k+2}{3}\right\rceil
  \leq j=\lceil\log k\rceil.
$$
This yields the simplified bound
$\Tarski(n,k)=O\!\left(5^k(\log n)^{\lceil\log k\rceil}\right)$,
completing the proof.
\end{proof}

\end{document}